\documentclass[sigconf, nonacm]{acmart}

\newcommand{\REV}[1]{#1}%{\color{blue}#1}}

\AtBeginDocument{%
  \providecommand\BibTeX{{%
    \normalfont B\kern-0.5em{\scshape i\kern-0.25em b}\kern-0.8em\TeX}}}

\usepackage{graphicx}
\usepackage{balance}  % for  \balance command ON LAST PAGE  (only there!)
\usepackage{adjustbox}
\usepackage{caption}
\usepackage{subcaption}
\usepackage{float}
\usepackage{hyperref}
\usepackage{tikz}
\usepackage{booktabs}
\usepackage{pgfplots}
\usepackage{colortbl}
\usepackage[noend,linesnumbered,ruled,vlined]{algorithm2e}
\usetikzlibrary{plotmarks}
\usetikzlibrary{matrix}
\usetikzlibrary{arrows,positioning} 
\usetikzlibrary{calc} 
\usetikzlibrary{through}
\usetikzlibrary{patterns}
\usetikzlibrary{decorations.pathmorphing}
\usetikzlibrary{decorations.pathreplacing}
\usepackage{mathrsfs,amsmath}
\usepackage{stackengine} 
\usepackage{ifsym}
\usepackage{multirow}
\usepackage{comment}
\usepackage{xcolor}

\newcommand{\para}[1]{\smallskip\noindent \textit{#1}}

\newcommand{\key}[1]{\, {\underline{{#1}}}\, }
\newcommand{\nonkey}[1]{#1}
\usetikzlibrary{arrows}

\usetikzlibrary{positioning}
\usepgfplotslibrary{fillbetween}

\DeclareMathOperator*{\argmin}{arg\,min}

\newcommand{\join}{\, \tiny \textifsym{|><|}\, }
\newcommand{\semijoin}{\, \tiny \textifsym{|><}\, }
\newcommand{\antijoin}{\, \tiny \textifsym{|>}\, }
\newcommand{\leftouter}{\, \tiny \textifsym{d|><|}\, }
\newcommand{\rightouter}{\, \tiny \textifsym{|><|d}\, }
\newcommand{\fullouter}{\, \tiny \textifsym{d|><|d}\, }

\pgfdeclarepatternformonly{south west lines}{\pgfqpoint{-0pt}{-0pt}}{\pgfqpoint{3pt}{3pt}}{\pgfqpoint{3pt}{3pt}}{
        \pgfsetlinewidth{0.4pt}
        \pgfpathmoveto{\pgfqpoint{0pt}{0pt}}
        \pgfpathlineto{\pgfqpoint{3pt}{3pt}}
        \pgfpathmoveto{\pgfqpoint{2.8pt}{-.2pt}}
        \pgfpathlineto{\pgfqpoint{3.2pt}{.2pt}}
        \pgfpathmoveto{\pgfqpoint{-.2pt}{2.8pt}}
        \pgfpathlineto{\pgfqpoint{.2pt}{3.2pt}}
        \pgfusepath{stroke}}

\begin{document}

\title{Weighted Random Sampling over Joins}

\begin{abstract}
Joining records with all other records that meet a linkage condition can result in an astronomically large number of combinations due to many-to-many relationships. For such challenging (acyclic) joins, a random sample over the join result is a practical alternative to working with the oversized join result. Whereas prior works are limited to uniform join sampling where each join row is assigned the same probability, the scope is extended in this work to weighted sampling to support emerging applications such as scientific discovery in observational data and privacy-preserving query answering. Notwithstanding some naive methods, this work presents the first approach for weighted random sampling from join results. Due to a lack of baselines, experiments over various join types and real-world data sets are conducted to show substantial memory savings and competitive performance with main-memory index-based approaches in the equal-probability setting. In contrast to existing uniform sampling approaches that require prepared structures that occupy contested resources to squeeze out slightly faster query-times, the proposed approaches exhibit qualities that are urgently needed in practice, namely reduced memory footprint, streaming operation, support for selections, outer joins, semi joins and anti joins and unequal-probability sampling. All pertinent code and data can be found at: https://github.com/shekelyan/weightedjoinsampling

\end{abstract}
%\begin{CCSXML}
%<ccs2012>
%  <concept_id>10010520.10010553.10010562</concept_id>
%  <concept_desc>Computer systems organization~Embedded systems</concept_desc>
%  <concept_significance>500</concept_significance>
%  <concept_id>10010520.10010575.10010755</concept_id>
%  <concept_desc>Computer systems organization~Redundancy</concept_desc>
%  <concept_significance>300</concept_significance>
%  <concept_id>10010520.10010553.10010554</concept_id>
%  <concept_desc>Computer systems organization~Robotics</concept_desc>
%  <concept_significance>100</concept_significance>
 % <concept_id>10003033.10003083.10003095</concept_id>
 % <concept_desc>Networks~Network reliability</concept_desc>
  %<concept_significance>100</concept_significance>
%</ccs2012>
%\end{CCSXML}

%\ccsdesc[500]{Computer systems organization~Embedded systems}

\author{Michael Shekelyan}
\authornote{ While at University of Warwick: development/analysis/evaluation of proposed approaches, writing/editing the manuscript and writing/running all the implemented code. While at King's College London: rewriting Section 3 and revising the manuscript.}
\email{michael.shekelyan@kcl.ac.uk}
\affiliation{%
  \institution{King's College London}
  %\streetaddress{Address}
  %\city{London}
  %\state{State}
  \country{UK}
  %\postcode{12345}
}

\author{Graham Cormode}
\email{g.cormode@warwick.ac.uk}
\affiliation{%
  \institution{University of Warwick}
  %\streetaddress{Address}
  \city{Coventry}
  %\state{State}
  \country{UK}
  %\postcode{12345}
}

\author{Peter Triantafillou, Ali Shanghooshabad, Qingzhi Ma }
\email{P.Triantafillou@warwick.ac.uk}
\affiliation{%
  %\institution{University of Warwick}
  %\streetaddress{Address}
  %\city{Coventry}
  %\state{State}
 \country{University of Warwick, UK}
  %\postcode{12345}
}

\renewcommand{\shortauthors}{Shekelyan et al.}

\begin{teaserfigure}
\centering
\includegraphics[width=0.9\textwidth]{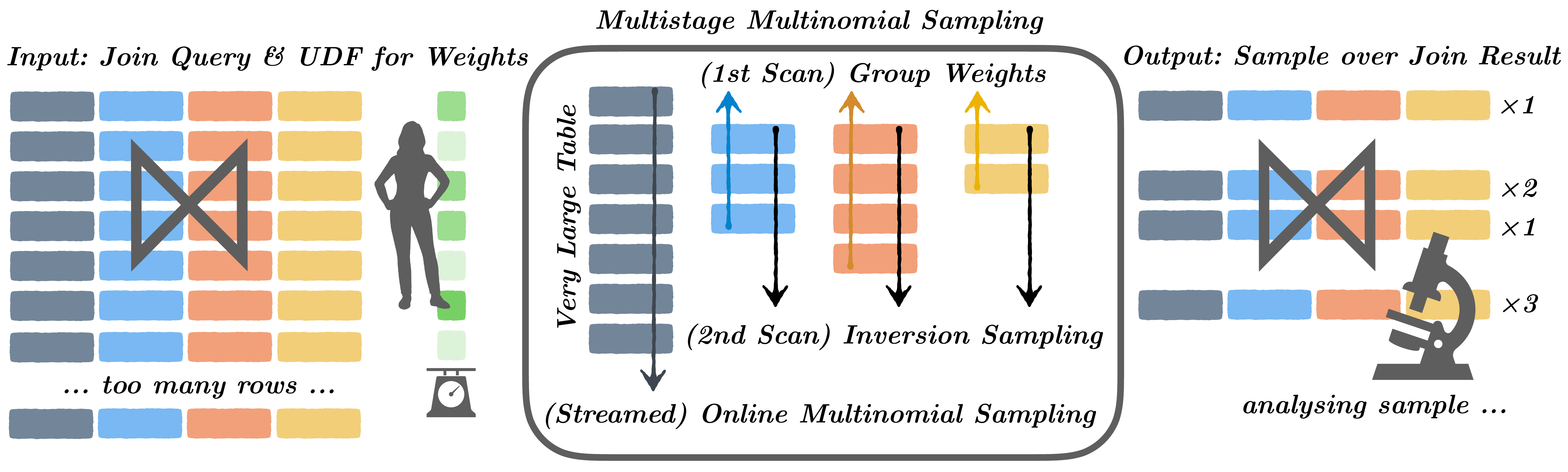}
\caption{The user provides a join query and weight function, the proposed stream sampler processes the largest table in one sequential scan and the other tables in two sequential scans and then outputs a weighted random sample over the join.
\label{fig:overview}}
\end{teaserfigure}

\maketitle

\section{Introduction}
\noindent

Collecting random samples has long been a vital tool to understand and analyse an underlying data population. The role of random sampling (unlike some quick-and-dirty heuristic sampling techniques \cite{DBLP:conf/sigmod/HaasH99, DBLP:conf/sigmod/0001WYZ16,feraud2010sampling, DBLP:conf/cidr/LeisRGK017}) 
is to provide a representative surrogate for the data population of interest with well-understood statistical properties that can be handed off directly to any systems or users.
%% SAMPLING IN DBS
Obtaining representative surrogates has become particularly pressing 
in the era of big data, where complex operations over massive data sets become prohibitively expensive.  Within the database literature, sampling is the key underpinning technology of approximate query processing (AQP) engines such as VerdictDB \cite{verdictdb} and QuickR \cite{quickr}.  Similarly, random samples are used for more advanced inferential analytics tasks, such as for training statistical and machine learning models, e.g., learned AQP engines depend on such samples \cite{deepdb,dbest++} and optimal subsampling  \cite{wang2021optimal,yao2019optimal,wang2021comparative,ting2018optimal,wang2019more} for tasks like logistic regression rely on unequal-probability samples.

\para{Why random sampling over joins?}
%
%% SAMPLING OVER JOINS
Joins of tables are a fundamental operator in relational database management systems, NoSQL databases, and in many analytics applications where data is spread across a variety of heterogeneous data sources. Unfortunately, multi-way joins can become excessively expensive, especially for very large tables, when join attributes are not keys and when they feature many-to-many relations. Fortunately, a random sample of the join result would be more than sufficient for many applications, for the reasons outlined above. The natural question that arises is then: Is sampling able to avoid the high costs of computing join queries? Although efficiently collecting random samples for a single base table is trivial, collecting a random sample over a join of tables remains a formidable challenge.

\para{Why weighted random sampling over joins?}

Weighted random sampling allows users and systems to weigh potential join results such that if a row in the join result has twice the weight of another, it is twice as likely to be included in the sample. Sampling with probabilities proportional to weights is not just some exotic variation of the sampling problem, but rather the inevitable response to a large variety of modern day requirements that cannot be meet with uniform probability sampling:

\begin{itemize}
\item \para{Privacy-Protection.}
In privacy-preservation frameworks for query-answering such as \emph{differential privacy} \cite{dwork2008differential} with various database applications \cite{dong2021residual,kotsogiannis2019privatesql,narayan2012djoin,machanavajjhala2008privacy}, any non-random selections based on user-derived selection predicates violate the formal privacy requirements. In contrast, weighted sampling satisfies the requirements as long as the weights are sufficiently robust to individual users \footnote{
\REV{
\emph{Differentially private selections over joins}: Suppose the joined tables are non-secret, while weights are sensitive, because they are based on private user information. In outline, suppose each column value $x$ has a utility value $u(x)$ s.t. each join row $x_1 x_2 \ldots x_n$ has a utility value $u(x_1)+u(x_2)+\ldots+u(x_n)$. Suppose the utility values are based on a set of private data records and that adding/removing a user's record from the set changes the sum $u(x_1)+u(x_2)+\ldots+u(x_n)$ by at most $\Delta$. The sampling weight is then determined using the exponential mechanism \cite{mcsherry2007mechanism}, i.e., $w(x_1 x_2 \ldots x_n) = \exp( \frac{\varepsilon}{2\Delta} \sum_{i=1}^n u(x_i))$. Generating a multinomial join sample of size $k$ is then $(k \varepsilon)$-differentially private via sequential composition \cite{dwork2014algorithmic}, which means the probability of observing a particular join sample increases/decreases at most by a multiplicative factor $\exp(k\varepsilon)$ due to the presence/absence of an individual user record.}}

\item \para{Stratified Sampling / Join over Selections.}
Join rows violating column-wise selection predicates can be assigned a weight of zero via column-based weights. Hence, weighted sampling can be used to collect stratified samples, e.g., as needed by AQP systems \cite{verdictdb,quickr,chaudhuri2007optimized}.

\item \para{Probability proportional to size (PPS) sampling.}
A common sampling design samples larger groups proportionally more likely \cite{rosen1997sampling,olken1995random,srivastava2004memory} which is needed for sampling over groups of records.

\item \para{Data exploration.}
Weighted sampling allows smoother types of selection as a function of data values with no hard cut-off, which allows to prioritise more relevant entries, e.g., featuring larger sales volumes, more recent sales or sales closer to a location. Similarly, it makes it easier to sample outliers as instead of having to know the parameters of an outlier, it suffices to weigh expected entries lower. Thus, weighted sampling is utilised for scientific discovery in observational data \cite{sidirourgos2013scientific}.

%\item \para{Bias removal.}
%In many cases the data is skewed with regards to what the users are interested in. 
%For instance, if the data is only based on customers from a certain region, but this data is meant to offer insights beyond the region, it is desirable to counteract the region-specific biases during the sampling, assuming that such biases are known and well understood.

\item \para{Aggregations.}
Weighted sampling can be more accurate for aggregations on attribute values. For instance, a random sample biased towards larger values is more accurate for aggregation functions such as SUM, since they have a bigger impact on the SUM aggregate. Unequal sampling probabilities are widely used for statistical estimators \cite{sampford1967sampling,brewer2013sampling}.

\end{itemize}

\para{How are the weights expressed?}
It is assumed in this work that the join result weights come from functions provided by the user that typically map each column value to a weight and the total weight is then the product of these column value weights. For instance, for a join of a table that keeps track of sales with a table that details everything about the sold items, the user can decide to weigh the sample based on number of sales multiplied by the price of the items. While the proposed approach also supports the less common case, where rows of base tables are given weights that are not products of column-based weights, it does not go as far as supporting arbitrary weights for join rows. 

\subsection{Contributions.} 
An overview of the proposed approach can be found in Figure~\ref{fig:overview}. Weighted random sampling over the join rows is formalised through a \emph{multinomial sampling design} \cite{kesten1959property,bloch1967bayesian,ruiz2008bayesian,emigh1983number,hoeffding1965asymptotically,cressie1984multinomial}, which samples with-replacement with probabilities proportional to weights. This means each sampled join row is sampled independently and there can be repetitions of join rows.  In order to compute aggregate weights of by-table-groups the problem is formulated as a graph problem that can be solved with an approach akin to dynamic programming. Unlike prior work \cite{DBLP:conf/sigmod/ChaudhuriMN99,shanghooshabad2021pgmjoins,DBLP:conf/sigmod/ZhaoC0HY18}, the proposed framework for multi-way joins (cf. Section~\ref{sec:main}) is sufficiently expressive to also support outer joins and semi/anti-joins and only requires stream-based access to the tables. The aggregate weights are then plugged into a multinomial sampler. For this purpose an online multinomial sampler is proposed (cf. Section~\ref{sec:general}) to collect multiple sampled join rows at once without requiring multiple passes over the main table. The online multinomial sampler faces a unique set of problems and although some works deal with simpler instances of the problem \cite{jayaram2019weighted,startek2016asymptotically}, the problem has not been previously solved apart from some naive methods. Finally it is shown how to correctly and efficiently utilise statistical methods to check if a sample follows the intended distribution (cf. Section~\ref{sec:gof}), discuss how everything ties together with existing works (cf. Section~\ref{sec:related}) and conduct an extensive experimental evaluation (cf. Section~\ref{sec:experiments}).

The experimental findings are that the proposed sampling approaches require much less memory than building indexes, typically by a factor of 10.  
Even if one is fortunate enough to have indices already in place, then the proposed approach is actually not much slower. However, in the more likely case that pre-built indexes are not available, the proposed approaches are up to two times faster.

\subsection{Limitations.} 

As an open problem remain general cyclic joins that might be insurmountable due to their relationship with other hard problems such as triangle counting (see also \cite{DBLP:conf/icdt/Chen020}). Essentially, the state-of-the-art as in \cite{DBLP:conf/sigmod/ZhaoC0HY18,shanghooshabad2021pgmjoins} is to reduce the problem to acyclic joins and then purge rows violating the cyclic join. Clearly, this presupposes that a reasonable fraction of the rows from some acyclic join does not violate the cyclic join. Hence, if this gamble works out or not is completely data-dependent. For sake of completeness, strategies for cyclic joins are also laid out and evaluated (see Section~\ref{sec:cyclic}), but as there is only limited benefit to solving some handpicked problem instances that are unlikely to be representative, the focus of this work lies primarily on acyclic joins. Such joins are challenging on their own in case of many-to-many relations (cf. Figure~\ref{fig:manytomany}) that \REV{can} lead to a vast increase in the join size (cf. Table~\ref{tab:joinsizes}).%, especially if it is attempted to avoid costly index structures that hamper the applicability of sampling methods in highly dynamic and ad-hoc settings.

\subsection{Background on Join Sampling} \label{sec:joinsamplingbackground}
The naive approach to first compute the join and then sample from the join result (either using inversion sampling or weighted reservoir sampling \cite{DBLP:journals/ipl/EfraimidisS06}) largely defeats the purpose.  To first sample and then join only produces a (simple) random sample if all join attribute values are distinct, but even then it is typically not useful as the resulting sample is too small (see also \cite{DBLP:conf/sigmod/ChaudhuriMN99, DBLP:conf/sigmod/AcharyaGPR99a}). 

Inversion sampling (popularised by Smirnov \cite{smirnov1939estimation} over 80 years ago) is as simple as uniformly drawing a random number between $0$ and $1$ and plugging it into the inverse of the cumulative distribution function (cf. Figure~\ref{fig:inverse}). The uniformly drawn random number is used as a normalised rank and the join row at that rank is retrieved using the inverse function.
While for joins with just a single join attribute one could sort all base tables by that attribute and systematically jump through the join results, it is not clear if this could be done for multi-way joins. 

Rejection sampling (popularised almost 70 years ago by Von Neumann~\cite{von195113}) modifies the ratios between selection probabilities by rejecting each element with an individual probability. For instance, suppose one element is twice as likely as the other one, but one would like to reverse this ratio, one could always accept the initially less likely one and reject the other with probability $3/4$. Similarly, if one wants to go from simple random sampling to weighted random sampling, one could always accept the element with the maximal weight and accept an element with half the maximal weight half the time. The acceptance rate is then roughly the ratio between average and maximal weight, which is often very small. %As finding the maximal weight requires inspecting all elements, one can instead use an upper bound, which further decreases the acceptance rate by the ratio between upper bound and maximal weight.

Join synopses (introduced by Achyara et al. \cite{DBLP:conf/sigmod/AcharyaGPR99a} over 20 years ago) focus on so-called foreign key joins where one can order the tables, such that each row in the first table joins with exactly one row in the second table and so on. Thus, one can just sample from the first table of that order and follow the links (see also Linked Bernoulli Synopses \cite{DBLP:conf/ssdbm/GemullaRL08}). As user-defined weights (or even selections) can destroy these links and become known at query time, this approach would need to be used with very large sample sizes and augmented with rejection sampling to utilise it in the weighted setting.

Zhao et al's  recent approach \cite{DBLP:conf/sigmod/ZhaoC0HY18} has been the first method for random sampling over multi-way joins. 
The basic idea is to start out with an approximate distribution (or the exact distribution as in the ``exact weights'' method) generalising ideas from Chaudhuri \cite{DBLP:conf/sigmod/ChaudhuriMN99} and then applying rejection sampling as in Olken's method \cite{olken1993random} to obtain the correct distribution. 
While this generalises both Chaudhuri's and Olken's method into a single framework for multi-way joins, the approach heavily depends on indexes for join attributes, solely focusing on sampling times, while other key practical considerations (such as memory usage and index availability and building time/space costs) are overlooked. The single other approach is \cite{shanghooshabad2021pgmjoins} which interestingly links the problem to probabilistic graphical models, but is again limited to uniform weights with a singular focus on query-time.
Critically, \REV{as published} these methods are limited to simple random sampling.

Another important concept is \textit{superset sampling}. 
In this context the superset would contain all join results, but also  additional cross product rows that violate some join conditions. As non-join rows are dropped later-on, the final probability of selecting them is zero and any proportions between join row probabilities are maintained.  For instance, the results of a cyclic join are a subset of results of an acyclic join that satisfies an additional join condition. As pointed out in \cite{DBLP:conf/sigmod/ZhaoC0HY18}, one can sample from a superset acyclic join and then drop any non-cyclic results. 
In this work, this idea is used for acyclic joins to reduce the memory footprint. If join conditions over equal attribute values are replaced with join conditions over equal hash values of attribute values, one can reduce the bookkeeping from per-attribute-value to per-hash-value and vastly reduce the memory footprint.

\tikzset{childarrow/.style={->}}
\tikzset{fkarrow/.style={red, dashed, thick}}

\tikzset{root/.style={rectangle, draw=black,fill=yellow!20, minimum width=1cm}}
\tikzset{child/.style={rounded corners, rectangle, draw=black, node distance=1.25cm}}

\tikzset{table/.style={matrix of nodes,align=left,inner sep=0.05cm, column sep=0cm, row sep=0cm, draw=black,
column 1/.style={anchor=base west},
    column 2/.style={anchor=base west},
    column 3/.style={anchor=base west}
},}

\begin{figure}
\centering
\begin{minipage}{0.22\textwidth}
\subcaption{Many-to-One}
\centering
\begin{tikzpicture}[scale=0.8]
\node[root](abc) {$AB$};
\node[child](bf) [right of=abc]{$\key{B}C$};
\draw[childarrow] (abc) --(bf) ;
\node[child](fg) [right of=bf]{$\key{C}D$};
\draw[childarrow] (bf) --  (fg);
\draw[fkarrow] ($(abc)+(0,-0.17)$) -- ($(abc)+(0,-0.4)$)  -- ($(bf)+(-0.12,-0.4)$) edge[->] ($(bf)+(-0.12,-0.17)$);
\draw[fkarrow] ($(bf)+(0.2,0.17)$) -- ($(bf)+(0.2,0.4)$)  -- ($(fg)+(-0.12,0.4)$) edge[->] ($(fg)+(-0.12,0.17)$);
\end{tikzpicture} 
$$ \pi_{A,B}(AB \bowtie BC \bowtie CD) = AB$$
\end{minipage}
\begin{minipage}{0.22\textwidth}
\subcaption{Many-to-Zero/One}
\centering
\begin{tikzpicture}[scale=0.8]
\node[root](abc) {$AB$};
\node[child](bf) [right of=abc]{$\key{B}C$};
\draw[childarrow] (abc)--  (bf);%node[below]{\foreignkey}
\node[child](fg) [right of=bf]{$\key{C}D$};
\draw[childarrow] (bf) --  (fg); %node[below]{\foreignkey}
\end{tikzpicture}
$$ \pi_{A,B}(AB \bowtie BC \bowtie CD) \subseteq AB$$
\end{minipage}
\begin{minipage}{0.22\textwidth}
\subcaption{Many-to-Many with Keys}
\centering
\begin{tikzpicture}[scale=0.8]
\node[root](abc) {$AB$};
\node[child](bf) [right of=abc]{$\key{B}C$};
\draw[childarrow] (abc)--  (bf);%node[below]{\foreignkey}
\node[child](fg) [right of=bf]{$CD$};
\draw[childarrow] (bf) --  (fg); %node[below]{\foreignkey}
\draw[fkarrow] ($(abc)+(0,-0.17)$) -- ($(abc)+(0,-0.4)$)  -- ($(bf)+(-0.12,-0.4)$) edge[->] ($(bf)+(-0.12,-0.17)$);
\end{tikzpicture}
%Many-to-One if $B,C$ is a key in $(BC \bowtie CD)$
\end{minipage}
\begin{minipage}{0.22\textwidth}
\subcaption{Many-to-Many}
\centering
\begin{tikzpicture}[scale=0.8]
\node[root](abc) {$AB$};
\node[child](bf) [right of=abc]{$BC$};
\draw[childarrow] (abc)--  (bf);%node[below]{\foreignkey}
\node[child](fg) [right of=bf]{$CD$};
\draw[childarrow] (bf) --  (fg); %node[below]{\foreignkey}
\end{tikzpicture} 
\end{minipage}
\caption{Taxonomy of different types of joins \label{fig:manytomany}}
\end{figure}
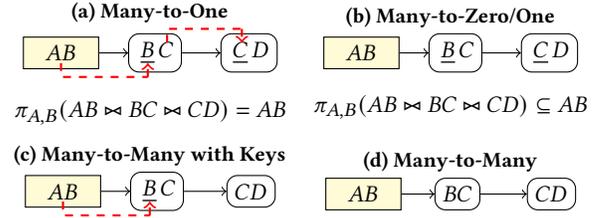

\section{Problem Setting}
\label{sec:problem}

\begin{figure}[t]
\begin{tikzpicture}[scale=0.9]
\input{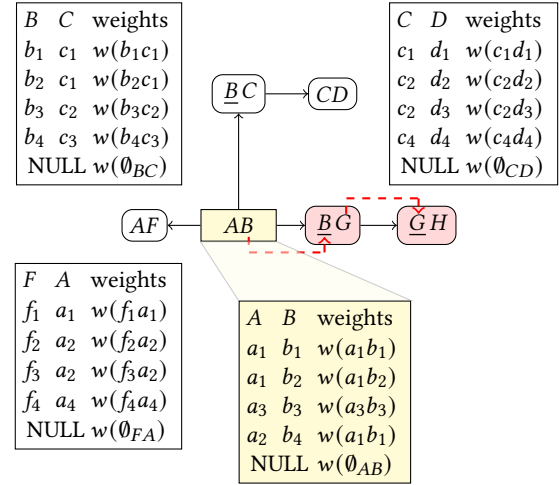}
\end{tikzpicture}
    \caption{Running example: Acyclic six-way join $(FA \fullouter AB \rightouter BC \semijoin CD) \join AG \join GH$ where $B$ from $AB$'s $B$ is a foreign key in $BG$ and $BG$'s $G$ is a foreign key in $GH$. $BC$'s $B$ is simply a key.
    \label{fig:sevenway}}%$AB \bowtie AF \bowtie BC \bowtie CD \bowtie AG \bowtie GH$
\end{figure}
\begin{figure}
\centering
\begin{tikzpicture}[scale=3.8]

\tikzstyle{eventline}=[black!60,ultra thin];

%\draw[draw=none,fill=none] (-0.5,-0,5) rectangle (1.5,1.5);

\draw (-0.2, 0) -- (-0.25, 0) node[anchor=east] {$0$};
\draw (-0.2, 1) -- (-0.25, 1) node[anchor=east] {$1$};

%\draw (0, 0) -- (-0.025, 0) node[anchor=north east] {\scriptsize$0.0$};

\draw (0, 0.0) -- (-0.025, 0.0) node[anchor=east] {\tiny $0.0$};
\draw (0, 0.2) -- (-0.025, 0.2) node[anchor=east] {\tiny$0.2$};
\draw (0, 0.4) -- (-0.025, 0.4) node[anchor=east] {\tiny$0.4$};
\draw (0, 0.6) -- (-0.025, 0.6) node[anchor=east] {\tiny$0.6$};
\draw (0, 0.8) -- (-0.025, 0.8) node[anchor=east] {\tiny$0.8$};
\draw (0, 1.0) -- (-0.025, 1.0) node[anchor=east] {\tiny$1.0$};
%\draw (0, 1.0) -- (-0.025, 1.0) node[anchor=north east] {\scriptsize$1.0$};

\draw (-0.1, 1.1) node[anchor=west] { $F(a) = \sum_{i = 1}^{a} p_i$};

%\draw (0.5, 1) node[anchor=south] {CDF};

%\draw (1, 0) node[anchor=west] { $a$};

\newcommand{\xxa}{0.0}
\newcommand{\xxb}{0.1}
\newcommand{\xxc}{0.2}
\newcommand{\xxd}{0.3}
\newcommand{\xxe}{0.4}
\newcommand{\xxf}{0.5}
\newcommand{\xxg}{0.6}
\newcommand{\xxh}{0.7}
\newcommand{\xxi}{0.8}
\newcommand{\xxj}{0.9}
\newcommand{\xxl}{1.0}

\newcommand{\yya}{0.05}
\newcommand{\yyb}{0.2}
\newcommand{\yyc}{0.4}
\newcommand{\yyd}{0.45}
\newcommand{\yye}{0.475}
\newcommand{\yyf}{0.5}
\newcommand{\yyg}{0.6}
\newcommand{\yyh}{0.8}
\newcommand{\yyi}{0.95}
\newcommand{\yyj}{0.978}
\newcommand{\yyl}{1.0}

\newcommand{\xxs}{1.0}
\newcommand{\xxt}{1.2}
%\draw (\xxs,0) --(\xxs,1);
%\draw (\xxt,0) --(\xxt,1);

\draw[draw=none] (\xxt,1.1*1/11) rectangle (\xxt+0.35, 1.1*0/11-0.1) node[pos=0.5]{\emph{events}};

%\draw[fill=white, draw=black!40, ultra thin] (\xxs,0) -- (\xxt, 0.1+1.1*0) -- (\xxt, 0.1+1.1*0/11) -- (\xxt+0.1, 0.1+1.1*0/11)--
%(\xxt+0.1, 0.1+1.1*11/11) -- (\xxt, 0.1+1.1*11/11)-- (\xxs, \yyl) -- cycle;

\draw[fill=yellow!10,draw=none] (\xxs,\yyg) -- (\xxt, 0.1+1.1*7/11) -- (\xxt, 0.1+1.1*7/11) -- (\xxt+0.35, 0.1+1.1*7/11)--
(\xxt+0.35, 0.1+1.1*8/11) -- (\xxt, 0.1+1.1*8/11)-- (\xxs, \yyh) -- cycle;

\draw[eventline] (\xxs,0) -- (\xxt, 0.1+1.1*0) -- (\xxt+0.35, 0.1+1.1*0);
\draw[eventline] (\xxs,\yya) -- (\xxt, 0.1+1.1*1/11)-- (\xxt+0.35, 0.1+1.1*1/11);
\draw[eventline] (\xxs,\yyb) -- (\xxt, 0.1+1.1*2/11) -- (\xxt+0.35, 0.1+1.1*2/11);
\draw[eventline] (\xxs,\yyc) -- (\xxt, 0.1+1.1*3/11)-- (\xxt+0.35, 0.1+1.1*3/11);
\draw[eventline] (\xxs,\yyd) -- (\xxt, 0.1+1.1*4/11)-- (\xxt+0.35, 0.1+1.1*4/11);
\draw[eventline] (\xxs,\yye) -- (\xxt, 0.1+1.1*5/11)-- (\xxt+0.35, 0.1+1.1*5/11);
\draw[eventline] (\xxs,\yyf) -- (\xxt, 0.1+1.1*6/11)-- (\xxt+0.35, 0.1+1.1*6/11);
\draw[eventline] (\xxs,\yyg) -- (\xxt, 0.1+1.1*7/11)-- (\xxt+0.35, 0.1+1.1*7/11);
\draw[eventline] (\xxs,\yyh) -- (\xxt, 0.1+1.1*8/11)-- (\xxt+0.35, 0.1+1.1*8/11);
\draw[eventline] (\xxs,\yyi) -- (\xxt, 0.1+1.1*9/11)-- (\xxt+0.35, 0.1+1.1*9/11);
\draw[eventline] (\xxs,\yyj) -- (\xxt, 0.1+1.1*10/11)-- (\xxt+0.35, 0.1+1.1*10/11);
\draw[eventline] (\xxs,\yyl) -- (\xxt, 0.1+1.1*11/11)-- (\xxt+0.35, 0.1+1.1*11/11);

\draw[draw=none, gray] (\xxt,0.1+1.1*1/11) rectangle (\xxt+0.1, 0.1+1.1*0/11) node[pos=0.5, anchor=west]{\scriptsize$X = \beta_1$};
\draw[draw=none, gray] (\xxt,0.1+1.1*2/11) rectangle (\xxt+0.1, 0.1+1.1*1/11) node[pos=0.5, anchor=west]{\scriptsize$X = \beta_2$};
\draw[draw=none, gray] (\xxt,0.1+1.1*3/11) rectangle (\xxt+0.1, 0.1+1.1*2/11) node[pos=0.5, anchor=west]{\scriptsize$X = \beta_3$};
\draw[draw=none, gray] (\xxt,0.1+1.1*4/11) rectangle (\xxt+0.1, 0.1+1.1*3/11) node[pos=0.5, anchor=west]{\scriptsize$X = \beta_4$};
\draw[draw=none, gray] (\xxt,0.1+1.1*5/11) rectangle (\xxt+0.1, 0.1+1.1*4/11) node[pos=0.5, anchor=west]{\scriptsize$X = \beta_5$};
\draw[draw=none, gray] (\xxt,0.1+1.1*6/11) rectangle (\xxt+0.1, 0.1+1.1*5/11) node[pos=0.5, anchor=west]{\scriptsize$X\ldots$};
\draw[draw=none, gray] (\xxt,0.1+1.1*7/11) rectangle (\xxt+0.1, 0.1+1.1*6/11) node[pos=0.5, anchor=west]{\scriptsize$X = \beta_{A-1}$};
\draw[draw=none] (\xxt,0.1+1.1*8/11) rectangle (\xxt+0.1, 0.1+1.1*7/11) node[pos=0.5, anchor=west]{\scriptsize$X = \beta_A$};
\draw[draw=none, gray] (\xxt,0.1+1.1*9/11) rectangle (\xxt+0.1, 0.1+1.1*8/11) node[pos=0.5, anchor=west]{\scriptsize$X = \beta_{A+1}$};
\draw[draw=none, gray] (\xxt,0.1+1.1*10/11) rectangle (\xxt+0.1, 0.1+1.1*9/11) node[pos=0.5, anchor=west]{\scriptsize$X\ldots$};
\draw[draw=none, gray] (\xxt,0.1+1.1*11/11) rectangle (\xxt+0.1, 0.1+1.1*10/11) node[pos=0.5, anchor=west]{\scriptsize$X = \beta_N$};

\draw[fill=black!5] (\xxa,\yya) -- (\xxb, \yya) -- (\xxb, \yyb) --  (\xxc, \yyb) -- (\xxc, \yyc) -- (\xxd, \yyc)  -- (\xxd, \yyd)-- (\xxe, \yyd) -- (\xxe, \yye) -- (\xxf, \yye) -- (\xxf, \yyf) -- (\xxg, \yyf) -- (\xxg, \yyg) -- (\xxh, \yyg) -- (\xxh, \yyh) -- (\xxi, \yyh) -- (\xxi, \yyi) -- (\xxj, \yyi) -- (\xxj, \yyj)  -- (\xxl, \yyj) -- (\xxl, \yyl) -- (1.0,1.0) -- (1.0, 0.0) -- (0.0, 0.0) -- cycle;

\draw[fill=black] (-0.2, 0.75) circle (0.02cm) node[anchor=east]{$u \sim U(0,1) \,$};

\draw[dashed] (-0.2, 0.75) -- (0.7, 0.75);
\draw[fill=black] (0.7, 0.75) circle(0.01cm);
\draw [thick, decorate,decoration={brace,amplitude=5pt},xshift=0.75pt,yshift=0pt]
(0.7,\yyh-0.01) -- (0.7,\yyg+0.01) node [black,midway,xshift=12pt] {\footnotesize
$p_A$};

\draw[thick] (0.7, 0.01) -- (0.7, -0.01);
\draw[dashed] (0.7, 0.75) -- (0.7, 0.0) node [anchor=north]{$X = {F}^{-1}(u)$};

%\draw [decorate,decoration={brace,amplitude=5pt},xshift=1pt,yshift=0pt]
%(0.7,0.69) -- (0.7,0.01) node [black,midway,xshift=25pt] {\scriptsize
%$F(n-1)$};

%\draw[fill=black] (0.7, 0) circle(0.01cm);

\draw (-0.2, 0) -- (-0.2, 1);

\draw[thick] (0,0) rectangle (1,1);

\end{tikzpicture}
\caption{\emph{Inversion sampling} draws from a distribution with distribution function $F$ by generating a random value $u \in [0,1]$ and then selecting $F^{-1}(u) = \argmin_x \{ F(x) | F(x) \ge u \}$. \label{fig:inverse}} 
\end{figure}
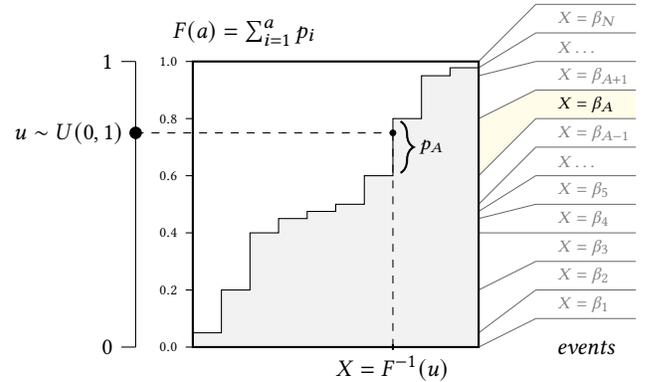

The running example throughout the paper can be found in Figure~\ref{fig:sevenway}, where the foreign-key appendages $BG$ and $GH$ are often omitted as they can be thought of as additional columns of $AB$. Since the example also features outer joins, null values are treated like special rows that can be assigned different weights for each table. The table $AB$ will serve as the main table throughout this work, which is typically picked as the largest table. In the proposed stream method (Section~\ref{sec:main}), only one sequential scan over the main table is needed, while two scans are needed for the other tables.

\subsection{Formal Definition}

\begin{definition}[user-defined weight function]\label{def:weights}
A user-defined weight function $w$ maps any base table row (vector) $\rho$ to a user-defined positive real $w(\rho)$ and for any combination of base table rows $w(\rho \times \tau \times \ldots \times \gamma) = w(\rho) \cdot w(\tau) \cdot \ldots \cdot w(\gamma)$.
\end{definition}

It is often convenient to express the weights over base table rows as products of per-column weights, which if not specified are simply set to $1$. Some care is required to not ``double-count''  join columns as they appear in multiple base tables. Selection operators over base tables can be expressed as weight functions, i.e., selected rows are mapped to $1$ and non-selected rows are mapped to $0$.

\begin{definition}[weighted random sample (with-replacement) over join queries]
A weighted sample (with-replacement) of size $n$ from a set of rows $\{\beta_1, \beta_2, \ldots, \beta_N\}$ with weighting function $w$ is a (sparsely represented) frequency vector $[x_1, x_2, \ldots, x_N]$ that follows a multinomial distribution with $n$ trials and normalised weights as event probabilities $p_1, p_2, \ldots, p_N$, i.e., $p_i = w(\beta_i) / \sum_{j = 1}^N w(\beta_j)$.
\end{definition}
 
 The rows $\{\beta_1, \beta_2, \ldots, \beta_N\}$ are in this case the result of the join query. If there are just two rows, the multinomial distribution has just two events and becomes simply a binomial distribution. Simple random sampling is a special case with $p_i = 1/N$. \REV{An acyclic/cyclic join query is in this work defined via a graph where each node is a table and each edge indicates that two tables are joined. In case of an acyclic join query, the join graph is a tree wheres in case of a cyclic join query, the join graph has a cycle (e.g., in Figure~\ref{fig:tpch} ``QX'' is acyclic, whereas ``QY'' is cyclic.)}.

\subsection{Desiderata}

\smallskip
\noindent
\textbf{Weighted:}
A \textit{weighted} sampler independently samples join rows with probability proportional to the product of arbitrary base table row weights.
\smallskip

Note that this definition is limited to weights that can be factorised over base tables, because sampling decisions otherwise would have to be postponed until a join row is completed. Arbitrary weights can still be achieved in post-processing through rejection sampling, but the effectiveness of the post-processing technique will depend on how close the existing probabilities come to the needed ones. One would expect this gap to be significantly decreased for factorised weights that aim to approximate the complex weighting functions.

\smallskip\noindent\textbf{Efficient:}
An efficient sampler performs for acyclic joins $O(T+n)$ operations and stores at most $O(T)$ values, where $T$ is the maximal table size and $n$ is the sample size.
\smallskip

As any sampler needs to process all base tables and generate all samples, $O(T+n)$ is optimal. 
While storing $O(T)$ values is not optimal, it is not an onerous requirement. 
The restriction to acyclic join queries is needed, as it is not known (and one might conjecture the opposite) if such a sampler can exist for cyclic join queries. 

\smallskip\noindent\textbf{Ad-hoc:}
An ad-hoc sampler behaves as if the data is generated at query time.
\smallskip

While the data likely exists in some form before query-time, in many cases it is only brought into the query-relevant form at  query time, as it either needs to be assembled or modified. 
For instance, different users might have different access to the data and some sensitive parts might have to be masked before it is handed over to the user. 
The ad-hoc property excludes precomputed indexes or other data preprocessing techniques. 
In this way, the ``ad hoc'' desideratum avoids the time/space costs related to building and storing indexes.

\smallskip\noindent\textbf{Streaming:}
A streaming sampler stores no previously seen rows completely unless they are sampled, accesses all tables at most twice as an unordered stream of rows and accesses one of the tables at most once.
\smallskip

The streaming setting typically implies just one pass over the data with insufficient memory to record the stream. When joining streams of rows, this is generally not possible as each new streamed row can turn any previous non-join row of another stream into a join row. For this reason, the single-pass condition only applies to one stream and allows other streams to be passed twice, which is optimal. Multi-pass streaming is still desirable as stream access is a universal interface and lends itself to embarrassingly parallel algorithms.

\newcommand{\YES}{\checkmark}
\newcommand{\NO}{}

\begin{table}
\caption{Desiderata Checklist \label{tab:desi}}
\begin{tabular}{rcccccc}
\toprule
Approach    & Weighted & Efficient & Ad-Hoc & Streaming \\
\midrule
\multicolumn{5}{c}{\emph{only two-way joins}} \\
\midrule
Olken'93 \cite{olken1993random}  & \YES & \NO & \YES & \YES \\
Chaudhuri'99 \cite{DBLP:conf/sigmod/ChaudhuriMN99}  & \NO & \YES & \NO & \YES \\
\midrule
\multicolumn{5}{c}{\emph{multi-way joins}} \\
\midrule
With-Index \cite{DBLP:conf/sigmod/ZhaoC0HY18,shanghooshabad2021pgmjoins} & \NO & \YES & \NO & \NO \\
Rejection \cite{von195113} & \YES & \NO & \YES & \NO \\
Inversion \cite{smirnov1939estimation} & \YES & \NO & \NO & \NO \\
\midrule
Proposed & \YES & \YES & \YES & \YES \\
\bottomrule
\end{tabular}
\end{table}

\medskip
In this work, a ``practical sampler'' is sought after that satisfies all these outlined properties, practical in the sense of a fully functional implementation with an empirical performance that is reflective of the asymptotic complexities, i.e., no large hidden constants.  As shown in Table~\ref{tab:desi} none of the existing works manages to satisfy all of these properties even for equal probabilities. \REV{Prior works do not support weighted sampling as published}\footnote{\REV{In the experimental section of this work one can find an implementation that modifies the proposed methods to use an index. While it comes fairly close to generalising \cite{DBLP:conf/sigmod/ZhaoC0HY18} to weighted sampling, it admittedly is more of a generalisation of proposed methods to using an index. Thus, it could be speculated that it should not be too hard to generalise prior works, but it has not been rigorously explored in this work.}}.  As mentioned, \cite{DBLP:conf/sigmod/ZhaoC0HY18,shanghooshabad2021pgmjoins} are efficient samplers, but assume memory-resident data with pre-computed indexes.

\newcommand{\texteachchild}{child $\mathcal{C}$ of $\mathcal{X}$ joined over $K$}
\newcommand{\eachchild}{\, child \, \mathcal{C} \, joined \, over \, K \,}

\section{Stream Join Sampler} %Weighted Random Sampling Over Joins of Multi-Pass Streams
\label{sec:main}

\subsection{Introduction}

\REV{
The basic idea of join sampling approaches is to extend a tuple table-by-table. For a uniform sample, the extension \emph{must} be drawn with probability proportional to the size of the (sub)join of a tuple with the tables downstream \cite{DBLP:conf/sigmod/ChaudhuriMN99, DBLP:conf/sigmod/ZhaoC0HY18}. This insight carries over to weighted join sampling, where the probability must be proportional to the total weight of the subjoin. In \cite{DBLP:conf/sigmod/ZhaoC0HY18} a tuple-oriented approach is taken to satisfy this constraint, where tuples are individually extended table-by-table, which is reliant on random access and index structures. Note that despite the use of indices \cite{DBLP:conf/sigmod/ZhaoC0HY18} still requires time linear in the size of the tables to collect the full sample. Hence, this work instead proposes a table-oriented approach, i.e., tables are scanned sequentially one-by-one. Intuitively, this is achieved by processing all extensions by one single table in one go akin to bread-first-search, rather than extending each tuple by all tables akin to depth-first-search. The foundation of such an approach is a proposed multipartite graph formulation that adds support to weighted sampling, non-equi joins, outer joins, selections, semi- and anti-joins (none of which are supported in prior work \cite{shanghooshabad2021pgmjoins,DBLP:conf/sigmod/ChaudhuriMN99,olken1993random, DBLP:conf/sigmod/ZhaoC0HY18}).
}

%\subsection{Preliminaries}

\para{Notation.}
Throughout this section, table names are used that match the columns, e.g., if a table has columns $A$ and $B$ the table will be called $AB$. If two tables have the same column name and are joined together, the join conditions will be to enforce equality across all columns with the same name as in a natural join. For attribute values of table rows for a column instead the lower case $x_1,x_2,\ldots$ is used to denote various independent values of a column $X$. 
%x_j,x_k,x_l,x_m,x_n
\subsection{Multipartite Graph Formulation}

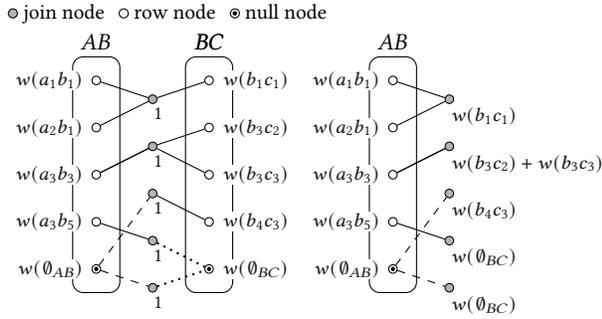
\begin{figure}
\begin{minipage}[b]{.94\linewidth}
{\small \begin{tikzpicture}\draw[black,fill=black!30] (0,0) circle [radius=0.06cm];\end{tikzpicture} join node \, \begin{tikzpicture}\draw[black] (0,0) circle [radius=0.06cm];\end{tikzpicture} row node \, \begin{tikzpicture}\draw[black] (0,0) circle [radius=0.06cm]; \fill[black] (0,0) circle [radius=0.03cm];\end{tikzpicture} null node }
\end{minipage}
\centering
\begin{minipage}[b]{.47\linewidth}
\centering \tikzstyle{op}= [matrix of nodes, inner sep=0.5cm,nodes = {style={ circle,fill=black!30,draw=black, inner sep=0.04cm}},column sep = 1.5cm, row sep = 0.75cm]
\tikzstyle{tab}= [rounded corners,op, draw=black, row sep = 0.75cm,nodes={style={ circle,fill=white,draw=black}}]
\tikzstyle{edge} = [color=black]
\tikzstyle{leftedge} = [color=black, dashed]
\tikzstyle{rightedge} = [color=black,  thick, dotted]
\tikzstyle{select} = [decorate, decoration={snake, amplitude=0.05cm}]
\tikzstyle{lab}= [, font=\footnotesize,node distance=0.1cm, anchor=north, pattern=crosshatch, pattern color=white, outer sep=0.001cm, inner sep=0.01cm]
\tikzstyle{wleft}= [, font={\footnotesize\color{black!90}},node distance=0.15cm, anchor=east, fill=white, outer sep=0.01cm, inner sep=0.01cm]
\tikzstyle{wright}= [, font={\footnotesize\color{black!90}},node distance=0.15cm, anchor=west, fill=white, outer sep=0.01cm, inner sep=0.01cm]
\tikzstyle{oplab}= [, font={\footnotesize\color{black!90}},node distance=0.1cm, anchor=north west, pattern=crosshatch, pattern color=white, outer sep=0.001cm, inner sep=0.01cm]
\tikzstyle{tabname}= [node distance=0.3cm, anchor=south]

%TOREDUCE

\begin{tikzpicture}[xscale=0.9]
\matrix (ab) [tab, row sep = 0.5cm, inner sep=0.25cm]{ {}\\ {}\\ {}\\ {}\\ {}\\ };
\matrix (b) [op, xshift=0.75cm, yshift=-0.25cm,row sep = 0.5cm]{ {}\\ {}\\ {}\\ {}\\ {}\\ };
\matrix (bc) [tab, xshift=1.5cm,row sep = 0.5cm, inner sep=0.25cm]{ {}\\ {}\\ {}\\ {}\\ {}\\ };

\fill[black] (ab-5-1) circle [radius=0.03cm];
\fill[black] (bc-5-1) circle [radius=0.03cm];

\draw[edge] (ab-1-1) -- (b-1-1);
\draw[edge]  (ab-2-1) -- (b-1-1);
\draw[edge]   (ab-3-1) -- (b-2-1);
\draw[edge]    (ab-3-1) -- (b-2-1);
\draw[leftedge]  (ab-5-1) -- (b-3-1);

\draw[leftedge]  (ab-5-1) -- (b-5-1);
\draw[edge]  (ab-4-1) -- (b-4-1);
\draw[edge]  (b-1-1) -- (bc-1-1);
\draw[edge]  (b-2-1) -- (bc-2-1);
\draw[edge]  (b-2-1) -- (bc-3-1);
\draw[edge]  (b-3-1) -- (bc-4-1);
\draw[rightedge]  (b-4-1) -- (bc-5-1);

\draw[rightedge]  (b-5-1) -- (bc-5-1);
\node[above of=ab-1-1,tabname]{$AB$};
\node[above of=bc-1-1,tabname]{$BC$};

\node[left of=ab-1-1, wleft]{$w(a_1b_1)$};
\node[left of=ab-2-1, wleft]{$w(a_2b_1)$};
\node[left of=ab-3-1, wleft]{$w(a_3b_3)$};
\node[left of=ab-4-1, wleft]{$w(a_3b_5)$};
\node[left of=ab-5-1, wleft]{$w(\emptyset_{AB})$};

\node[below of=b-1-1, oplab]{$1$};
\node[below of=b-5-1, oplab]{$1$};
\node[below of=b-2-1, oplab]{$1$};
\node[below of=b-3-1, oplab]{$1$};
\node[below of=b-4-1, oplab]{$1$};

\node[above of=bc-1-1,tabname]{$BC$};
\node[right of=bc-1-1, wright]{$w(b_1c_1)$};
\node[right of=bc-2-1, wright]{$w(b_3c_2)$};
\node[right of=bc-3-1, wright]{$w(b_3c_3)$};
\node[right of=bc-4-1, wright]{$w(b_4c_3)$};
\node[right of=bc-5-1, wright]{$w(\emptyset_{BC})$};

\end{tikzpicture}
\subcaption{Each join node is the root of a sub-tree (left to right). }\label{fig:graphjoin}
\end{minipage}
\begin{minipage}[b]{.47\linewidth}
\centering \tikzstyle{op}= [matrix of nodes, inner sep=0.5cm,nodes = {style={ circle,fill=black!30,draw=black, inner sep=0.04cm}},column sep = 1.5cm, row sep = 0.75cm]
\tikzstyle{tab}= [rounded corners,op, draw=black, row sep = 0.75cm,nodes={style={ circle,fill=white,draw=black}}]
\tikzstyle{edge} = [color=black]
\tikzstyle{leftedge} = [color=black, dashed]
\tikzstyle{rightedge} = [color=black,  thick, dotted]
\tikzstyle{select} = [decorate, decoration={snake, amplitude=0.05cm}]
\tikzstyle{lab}= [, font=\footnotesize,node distance=0.1cm, anchor=north, pattern=crosshatch, pattern color=white, outer sep=0.001cm, inner sep=0.01cm]
\tikzstyle{wleft}= [, font={\footnotesize\color{black!90}},node distance=0.15cm, anchor=east, fill=white, outer sep=0.01cm, inner sep=0.01cm]
\tikzstyle{wright}= [, font={\footnotesize\color{black!90}},node distance=0.15cm, anchor=west, fill=white, outer sep=0.01cm, inner sep=0.01cm]
\tikzstyle{oplab}= [, font={\footnotesize\color{black!90}},node distance=0.1cm, anchor=north west, pattern=crosshatch, pattern color=white, outer sep=0.001cm, inner sep=0.01cm]
\tikzstyle{tabname}= [node distance=0.3cm, anchor=south]
% REDUCED

\begin{tikzpicture}[xscale=0.9]
\matrix (ab) [tab, row sep = 0.5cm, inner sep=0.25cm]{ {}\\ {}\\ {}\\ {}\\ {}\\ };
\matrix (b) [op, xshift=0.75cm, yshift=-0.25cm,row sep = 0.5cm]{ {}\\ {}\\ {}\\ {}\\ {}\\ };

\fill[black] (ab-5-1) circle [radius=0.03cm];

\draw[edge] (ab-1-1) -- (b-1-1);
\draw[edge]  (ab-2-1) -- (b-1-1);
\draw[edge]   (ab-3-1) -- (b-2-1);
\draw[edge]    (ab-3-1) -- (b-2-1);
\draw[leftedge]  (ab-5-1) -- (b-3-1);
\draw[leftedge]  (ab-5-1) -- (b-5-1);
\draw[edge]  (ab-4-1) -- (b-4-1);
\node[above of=ab-1-1,tabname]{$AB$};

\node[left of=ab-1-1, wleft]{$w(a_1b_1)$};
\node[left of=ab-2-1, wleft]{$w(a_2b_1)$};
\node[left of=ab-3-1, wleft]{$w(a_3b_3)$};
\node[left of=ab-4-1, wleft]{$w(a_3b_5)$};
\node[left of=ab-5-1, wleft]{$w(\emptyset_{AB})$};

\node[below of=b-1-1, oplab]{$w(b_1 c_1)$};
\node[below of=b-5-1, oplab]{$w(\emptyset_{BC})$};
\node[below of=b-2-1, oplab]{$w(b_3 c_2)+w(b_3 c_3)$};
\node[below of=b-3-1, oplab]{$w( b_4c_3) $};
\node[below of=b-4-1, oplab]{$w(\emptyset_{BC})$};
\end{tikzpicture}
\subcaption{Weights of each sub-tree can be transferred into join node.}\label{fig:graphsemijoin}
\end{minipage}
\caption{ Graph formulation for two-way join\label{fig:graphformulation}}
\end{figure}

\begin{figure*}
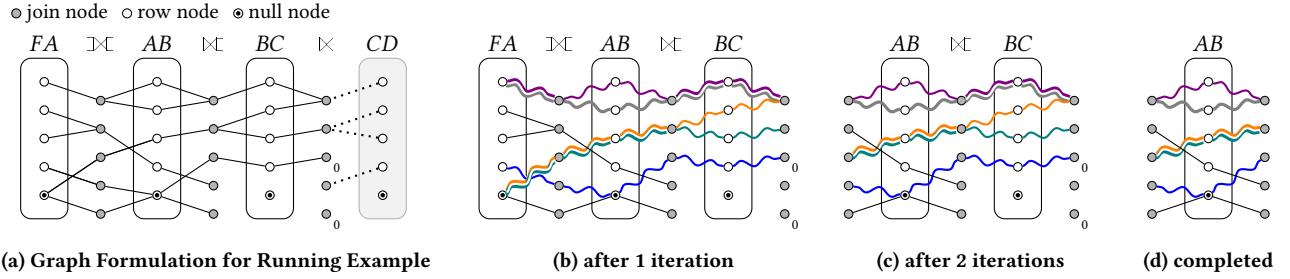

\begin{minipage}{0.95\linewidth}
{\small \begin{tikzpicture}\draw[black,fill=black!30] (0,0) circle [radius=0.06cm];\end{tikzpicture} join node \, \begin{tikzpicture}\draw[black] (0,0) circle [radius=0.06cm];\end{tikzpicture} row node \, \begin{tikzpicture}\draw[black] (0,0) circle [radius=0.06cm]; \fill[black] (0,0) circle [radius=0.03cm];\end{tikzpicture} null node }
\end{minipage}
\centering
\begin{minipage}[b]{.35\linewidth}
\centering \input{figures/multiway0}
\subcaption{Graph Formulation for Running Example}
\end{minipage}
\begin{minipage}[b]{.28\linewidth}
\centering \input{figures/multiway1}
\subcaption{after 1 iteration}
\end{minipage}
\begin{minipage}[b]{.20\linewidth}
\centering \input{figures/multiway2}
\subcaption{after 2 iterations}
\end{minipage}
\begin{minipage}[b]{.15\linewidth}
\centering \input{figures/multiway3}
\subcaption{completed \label{fig:completed}}
\end{minipage}
\caption{Algorithm~\ref{alg:weights}'s iterations (Lines~\ref{alg:iterbegin}-\ref{alg:iterend}) preserve the weight of result trees, but simplify the graph by removing partitions \label{fig:multiway} } 
\end{figure*}

\para{Graph formulation.} In order to more easily reason about not just inner joins, but also outer joins, it is useful to formulate the join operation through a graph, where rows and join attribute values form a multipartite graph with one partition of nodes per table and join column.

An example of such a multipartite graph for a two-way join can be found in Figure~\ref{fig:graphformulation}. The nodes of the graph are:

\begin{itemize}
\item \emph{Row nodes}: Each table is a set of nodes, one node for each row, e.g., there is a node for row $a_1 b_1$. 
\item \emph{Null nodes}: ``NULL'' values for a table are presumed to be an additional row of the table, e.g., $\emptyset_{AB}$.
\item \emph{Row node labels}: Each row node is labelled with the row weight, e.g., $w(a_1b_1)$, even null nodes, e.g., $w(\emptyset_{AB})$. 
\item \emph{Join nodes}: Each join column is a set of nodes, one node for each attribute value, e.g., $b_1$.
\item \emph{Join node labels}: Each join node is initially labelled with the value $1$, e.g., the node for $b_1$ has a label $1$.
\end{itemize}

As join operators are not necessarily symmetrical, there is a ``left'' and a ``right'' side. The edges of the graph potentially connect join nodes with row nodes with matching join attribute value, but it depends on the join operator if they do:

\begin{itemize}
\item Each join node for $\leftouter$,  $\rightouter$, $\fullouter$ and $\join$ is connected to each matching row node.
\item Each join node for $\rightouter$ and $\fullouter$ that is not connected to any row node on the left side is connected to the left null node.
\item Each join node for $\leftouter$ and $\fullouter$ that is not connected to any row node on the left side is connected to the right null node.
\item Each join node for $\semijoin$ has label $0$ if there is \emph{no} match on the right side. 
\item Each join node for $\antijoin$ has label $0$ if there \emph{is} a match on the right side.
\end{itemize}

This simply means left outer joins ($\leftouter$) allow null values on the left side and right outer joins ($\rightouter$) on the right side, while
full outer joins allow them on either side ($\fullouter$). Semijoins ($\semijoin$) and antijoins ($\antijoin$) are filters that change the weight of a join node depending on if there is a match on the right side.

The main motivation for this formulation, is that in such a graph, the trees spanning all partitions corresponds to join rows:

\begin{itemize}
\item \emph{Result trees}: Each tree that contains exactly one node from each reachable partition corresponds to a join row (right side partitions of semijoins or antijoins are unreachable). 
\item Trivial rule: A result tree cannot contain join nodes that are only connected to null nodes.
\end{itemize}

The weighted sampling problem can then be posed via the multipartite graph formulation:

\begin{definition}[Weighted Sampling over Joins]

Select a random result tree $R$ (corresponding to a join row) comprised of nodes with labels $w_1, \ldots, w_k$ with probabilities proportional to the tree's total weight $\prod_{i = 1}^k w_i$.

\end{definition}

As there are as many result trees as join rows, some method is needed to reduce the complexity noticeably below the number of join rows. As each tree is sampled with probabilities proportional to their weight, any group of trees will be sampled with probabilities proportional to their sum of weights. Thus, one can sample the result tree step by step rather than in one go, which will be shown in the following to allow for sub-linear complexity.

\para{Efficient solution to graph problem.}  The basic idea is to group all result trees that have the same root and then sample such a root-group. The challenge left to solve is how to compute the sum of weights of each root-group. The right side of Figure~\ref{fig:graphformulation} shows the most crucial primitive for this purpose. The operation sums the weights of all sub-trees of each join node and adds it as a new label of the join node. Such an operation can be implemented through a linear-time sequential scan of  $BC$ and generating a hash table of $B$-values contained in $BC$ where each entry holds the sum of observed weights of each $B$-value. In SQL-terms, one can think of such an operation as something along the lines of {\tt{Select SUM(W) from BC group by B}} where $W$ is presumed to be some column holding the weights of $BC$. As a second step, one can scan $AB$, which holds the root nodes. Now, for each row $a_ib_j$ of $AB$ one can look up the sum of sub-tree weights $W(b_j)$ in the hash table entry of the key $b_j$. The total weight of all groups that contain $a_ib_j$ is then $w(a_ib_j) W(b_j)$. If $BC$ is not a leaf node, one simply recursively continues this procedure to compute the weights of $BC$ until the leaf nodes are reached. Once a group with $a_ib_j$ as a root node is sampled with probability proportional to its total weight, one can continue in a similar fashion with the children nodes (adjacent nodes in the graph facing away from the main table), until the full tree is obtained. Semi-joins, anti-joins and selections can be supported through weights, which is detailed later.

\begin{algorithm}
\SetAlgoLined
\DontPrintSemicolon
Let $G = (V,E)$ be a multipartite graph corresponding to the join query and $w(\rho)$ be a node $\rho$'s label. \;
Let $\textsc{root}(V) = \mathcal{A} \subseteq V$ be the row node partition of the join query's main table. \label{alg:root} \; 
Let $\textsc{parent}_\mathcal{A}(\mathcal{T}, V')$ be all neighbouring join node partitions in $V'$ closer to $\mathcal{A}$ than $\mathcal{T}$. \label{alg:parent} \;
Let $\textsc{children}_\mathcal{A}(\mathcal{T}, V')$ be all neighbouring join node partitions in $V'$ further away from $\mathcal{A}$ than $\mathcal{T}$. \label{alg:children}\;
Let $\textsc{leafs}_\mathcal{A}(V')$ be all row node partitions in $V'$ furthest away from $\mathcal{A}$, excluding $\mathcal{A}$. \label{alg:leafs}\;

Initiate $V'$ as $V$. \;
%\tcc{Keep in implementation a hash table $H_{\mathcal{J}}$ for each partition of join nodes $\mathcal{J}$}
\While{ $|\textsc{leafs}_\mathcal{A}(V') | > 0$ \label{alg:iterbegin}}{

    Pick a row node partition $\mathcal{T} \in \textsc{leafs}_\mathcal{A}(V') $ \;
    Let $\mathcal{J}_0 = \textsc{parent}_\mathcal{A}(\mathcal{T}, V')$. \;
    %\tcc{Implement by setting default value of hash table $H_{\mathcal{J}_0}$ to $0$:}
    Set the label of all nodes in $\mathcal{J}_0$ to $-\infty$ \;
    
    %\tcc{Implement by scanning $\mathcal{T}$'s base table :}
    \ForEach{row node $\rho \in \mathcal{T}$ \label{alg:scanbegin}}{
    
        Initiate $W$ as $1$ \label{alg:weightbegin}\;
        \ForEach{join node part. $\mathcal{J} \in \textsc{children}_\mathcal{A}(\mathcal{T}, V')$}{
            %\tcc{Implement through look up in $H_{\mathcal{J}}$ using $\rho$'s join attribute value:}
            Let $W_\mathcal{J}$ be the label of $j \in \mathcal{J}$ connected to $\rho$ \;
            Multiply $W$ by $W_\mathcal{J}$.
        }\label{alg:weightend}
        
       % \tcc{Implement by adding/modifying entry in $H_{\mathcal{J}_0}$ using $\rho$'s join attribute value:}
        Let $W_{\mathcal{J}_0}$ be the label of $j \in \mathcal{J}_0$ connected to $\rho$ \label{alg:joinbegin}\;
        
        \lIf{ $W_{\mathcal{J}_0} = -\infty$ }{ Set $W_{\mathcal{J}_0}$ to $0$}
        
        Add $w(\rho) W$ to $W_{\mathcal{J}_0}$. \label{alg:joinend}\;
    } \label{alg:scanend}
     Remove row node partition $\mathcal{T}$ from $V'$ \label{alg:removebegin}\;
     \ForEach{join node partition $\mathcal{J} \in \textsc{children}_\mathcal{A}(\mathcal{T}, V')$}{
       Remove join node partition $\mathcal{J}$ from $V'$ \label{alg:removeend}\;
     }
     %\tcc{Implement by setting default value of $H_{\mathcal{J}_0}$ to $w(\emptyset_{\mathcal{T}})$:}
    Replace all labels $-\infty$ with $w(\emptyset_{\mathcal{T}})$ for nodes in $\mathcal{J}_0$ \;
}\label{alg:iterend}
\caption{Group Weights} \label{alg:weights}
\end{algorithm}

\para{Algorithm.}
It is described in Algorithm~\ref{alg:weights} how the group weights can be obtained for general multi-way joins. Figure~\ref{fig:multiway} applies the algorithm for the running example. $AB$ serves as the main table (see Line~\ref{alg:root}), such that the parent moves closer to $AB$ (Line~\ref{alg:parent}) and children node move further away from $AB$ (Line~\ref{alg:children}), while leafs such as $CD$ are on the outskirts furthest away from $AB$ (Line~\ref{alg:leafs}). The algorithm processes a new table in each iteration (Lines ~\ref{alg:iterbegin}-\ref{alg:iterend}). It makes a single stream pass over the rows of the new table (Lines ~\ref{alg:scanbegin}-\ref{alg:scanend}, computes the total weight of the sub-trees rooted at the row (Lines~\ref{alg:weightbegin}-\ref{alg:weightend}) and add this sub-tree weight multiplied by its own weight $w(\rho)$ to the parent join node (Lines~\ref{alg:joinbegin}-\ref{alg:joinend}). After the new table has been processed it is removed from further consideration (Lines~\ref{alg:removebegin}-\ref{alg:removeend}) and the algorithm goes back to Line~\ref{alg:iterbegin} and terminates when no more tables are left for consideration. After termination, each row node $\rho$ of the main table is linked to multiple join nodes and the product of the join node labels multiplied by $w(\rho)$ yields the total weight of all results trees / join rows containing $\rho$.

The algorithm can be implemented with a hash table $H_\mathcal{J}$ for each join node partition $\mathcal{J}$ that allows an efficient lookup of the linked join node for a row node based on the join attribute value. Note that if all join nodes except a few have the same label, one only keeps entries for the exceptions and maintains a default value for the rest.

\REV{

The hash table entries for a table can be computed in one scan that skips any rows that do not satisfy the selection predicates.

In case of a semi-join, the default value is $0$ and only entries are $1$ for results of the semi join. While an anti-join can be as large as a table, for sampling it can be supported via semi-join: For anti-joins the default value is $1$ and only entries are $0$ for results of the semi join. 

Theta/non-equi-joins can also be easily supported. If the link condition is $\neq$, then in addition to the hash-table one needs to maintain the total weight of all hash-table entries. Then the weight for equi-joins can simply be subtracted from the total weight, to obtain the $\neq$-join weight. For theta joins with a binary operator $\odot \in \{ <, \le, \ge, >\}$ for the link condition, one can first obtain the equi-join hash table and then replace the values with cumulatives. This means that the entry for each join attribute value $x$ holds the sum of weights of any equi-join entry $y$ that satisfies $y \odot x$. Additionally a binary search tree needs to be constructed to efficiently find the last value $y \odot x$ where $x$ is the queried value from a joining table. Then the queried value can be rewritten as the last hash table value.

It is also possible to obtain some limited support for Group by's, as they can be performed after the sample is collected. In order to obtain a better coverage of groups, the weights should ideally be chosen inversely proportional to the size of the groups.
}

\subsection{Multistage Multinomial Sampling}

In the previous section weighted join sampling has been mapped to the problem of sampling result trees with probabilities proportional to their weight (in simple cases they are paths). In a usual implementation, each join node partition is implemented as a hash map, such that by going each row $\rho$ in the main table one can look up matching entries in the hash maps and compute the product of the looked up values and $w(\rho)$ to obtain the total weight $W(\rho)$ of all result trees that contain $\rho$. After executing Algorithm~\ref{alg:weights} which brings the graph into the state of Figure~\ref{fig:completed}, the total weight $W(\rho)$ of each $\rho$ can therefore be computed using at most one hash-table look-up per table.

 As the task is to sample with probabilities proportional to $W(\rho)$, one can in a first stage perform a stream pass over the main table and collect a sample using the proposed online multinomial sampler from Section~\ref{sec:general}. After collecting the main table sample, this yields the sampled groups of result trees grouped by the main table row and it is left sampling within the groups. In the graph formulation (see Figure~\ref{fig:graphformulation}), each row $\rho$ in the main table sample is on the ``left'' side, linked to a join node in the ``middle'' that links to multiple row nodes on the ''right side. The task is then to sample for each sampled main table row, a row node on the ``right'' side with probabilities proportional to the rest of the result tree weight. Luckily, the total weights of the sub-trees have been previously computed and are readily available in the hash maps for the join nodes. As the total weight $W$ of the right side is known, one can use here inversion sampling (see Figure~\ref{fig:inverse}), i.e., draw a random number $u$ between $0$ and $W$ and go through the right side table until the total weight of observed rows is more than $u$ and then pick the preceding row. This allows to collect all sample continuations of the main table sample in one stream pass. Thus, in each stage the rows in the sample are extended by the row of another table until all tables that participate in the join have been reached.  The main table due to the online multinomial sampler from Section~\ref{sec:general} is only scanned once, while all other tables are scanned twice. In the general case, the second scans of the other tables cannot be avoided. Hence, this approach is optimal with regards to the number of sequential scans.

\subsection{Handling Cyclic Joins} \label{sec:cyclic}

While the main focus of this paper are the more common acyclic joins, this section discusses the more exotic cyclic joins that sometimes receive some academic interest. As in the literature, cyclic joins are treated as a selection predicate over the join rows of an acyclic join.

For a join graph where each node is a table and each edge is a join condition, a join query is called cyclic
if the join graph features at least one cyclic path that passes through at least two edges and three nodes.

\para{Rewriting as Selection over Acyclic Join Query.}
Any cyclic join query such as $AB \bowtie BC \bowtie CA$ can be rewritten into a selection over an acyclic join, e.g., {$\sigma_{A_1 = A_2} A_1B \bowtie BC \bowtie CA_2$}. 
If a join query has multiple cycles, each cycle adds one equality predicate to the selection operator. Note that here a natural join formalisation is used and  join conditions are removed by renaming attributes.

If the selection has high selectivity, i.e., selects most of the records, then one can first sample from $A_1B \bowtie BC \bowtie CA_2$ and then drop any samples that do not satisfy the selection predicate $\sigma_{A_1 = A_2}$. The acceptance rate is then on expectation equal to the selectivity, i.e., ${|\sigma_{A_1 = A_2} A_1B \bowtie BC \bowtie CA_2|}\ /\ {|A_1B \bowtie BC \bowtie CA_2|}$. As the acceptance rate can vary wildly depending on which join conditions are outsourced as selections, this leads to an interesting problem of an optimal rewriting, i.e., one that maximises the selectivity. 
A simple heuristic to select which edge to remove from the cycle is to pick the join link $X \bowtie Y$ that maximises the probability
$|X \bowtie Y|/(|X| \cdot |Y|)$ that two independently drawn random records from both tables are linked. This is in principle very similar to the Chow-Liu algorithm \cite{DBLP:journals/tit/ChowL68} that aims to break links where the mutual information is lowest.

The rewriting of a cyclic join query into a selection over an acyclic join is described in the following. One goes through each node (table) and follows each of its adjacent edges (join conditions) to a neighbouring node. For each adjacent edge one looks for the shortest path between its two end nodes that is not allowed to pass the adjacent edge. If no such path can be found then the node does not participate in a cycle. If such a path can be found, the node participates in at least one cycle and one can pick an edge in that cycle that one wants to outsource as a selection predicate. After the removal of an edge, one can check for cycles again until the node no longer participates in a cycle. For a join graph with $|V|$ nodes and $|E|$ edges such a procedure performs at most ${|V|}^2 \mathcal{C}$ shortest path searches where $\mathcal{C}$ is the number of cycles. As join graphs are very small and the number of cycles is at most $2^{|V|^2-|V|}$, the searches only takes a couple of milliseconds to perform. A more challenging problem is how to pick which edge to remove. While one could estimate the selectivity for each candidate, that can be quite expensive and take more time than the subsequent sampling step.

%Note that \cite{DBLP:conf/sigmod/ZhaoC0HY18} uses a rejection sampling approach instead of rewriting the join query, which introduces some additional rejections and only supports removing a connected subgraph. Thus, it is less clean and not as flexible.

\section{Economical Join Sampler}
\label{sec:economic}

The previous section presented the fundamentals for stream sampler.
It will now be considered how one can obtain an economic sampler with a reduced memory footprint for various types of joins.
%However, there are two further interesting cases that need be handled: foreign key joins, cyclic joins and joins with high-cardinality join attributes.

\subsection{Exploiting Foreign Key Joins}

Many-to-one relations allow to first sample from the ``many'' table and then look-up the ``one'' entity in the other table(s). If the weights are all equal, then this reduces to sampling from one table and then joining the sample with the other tables. 
If the weights are not equal, one can either first proceed as if they were equal and afterwards employ rejection sampling to rectify inclusion probabilities, or one needs to find the group weights and treat the foreign-key join like a many-to-many join. 
The former is more memory efficient, whereas the latter is reliably fast. 
For the economical sampler one picks the more memory efficient variant.

\subsection{Simplifying Cyclic Joins}

While some join queries can look like a challenging cyclic join on first glance, it is sometimes possible to reduce it to a much simpler query, e.g., an acyclic one.  
If a subgraph of the join graph is a foreign-key join, one can efficiently join the subgraph together and the join result can only have as many rows as the largest of the joined tables. To automate and generalise this process without knowledge of foreign key constraints, one can simply join all tables together where the join result is at most slightly larger than the tables itself.  For table sizes up to $N$, this simplification step requires $O(N)$ expected time and $O(N)$ space using hash joins or $O( N \log N)$ worst-case time
and $O(N)$ space using sort-merge joins.

\subsection{Novel Hashed Join Method: Handling High-Cardinality Join Attributes without Approximations} 
In case there are many distinct attribute values over join attributes within very large base tables, additional techniques are needed to reduce the memory footprint. A good way to think about join rows, is as a subset of cross product rows that satisfy join conditions. By relaxing the equi-join condition to a equi-hash-join conditions, one obtains a superset:

\begin{definition}[Equi-Hash Join]
The equi-hash join along some join column between two tables $AB$ and $BC$ is the subset of the crossproduct of rows where
the hash value in the join column matches for a shared hash function.
\end{definition}

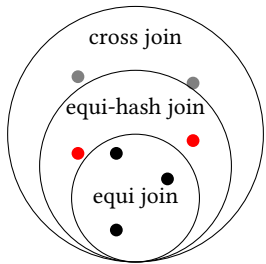
\begin{figure}
\begin{tikzpicture}[scale=0.85]
\draw (0,0) circle (1cm);

\fill[color=black] (0.5,0.3) circle (0.1cm);
\fill[color=black] (-0.3,-0.5) circle (0.1cm);
\fill[color=black] (-0.3,0.7) circle (0.1cm);

\fill[color=red] (-0.9,0.7) circle (0.1cm);
\fill[color=red] (0.9,0.9) circle (0.1cm);
\fill[color=gray] (-0.9,1.9) circle (0.1cm);

\fill[color=gray] (0.9,1.8) circle (0.1cm);

\draw (0,0) node[anchor=center] {equi join};
\draw (0,0.5) circle (1.5cm);
\draw (0,1.4) node[anchor=center] {equi-hash join};
\draw (0,1) circle (2cm);
\draw (0,2.5) node[anchor=center] {cross join};
\end{tikzpicture}
\caption{Hierarchy of join results. Samples from any equi-hash join with superflous elements purged (depicted in red) are ordinary samples from the equi join \label{fig:joinhierarchy}}
\end{figure}

Note that any result in the equi-join also must be in the equi-hash join, because if two attribute values are equal, so are their hashes (see also Figure~\ref{fig:joinhierarchy}\REV{)}. Thus, equi-hash joins reduce the number of join attribute values in a controlled way (controlled by the chosen hash function) by allowing to operate over the smaller number of hash values.  While the reduced domain leads to collisions, all those collisions only lead to some extraneous results that can be purged later on without any negative impact on precision. More specifically, the collisions cause additional rows from the cross product that violate the join conditions\REV{,} but are easy to identify during the post-processing of the sample. %Note that sampling from a superset, but then only keeping elements from subset, is equivalent to sampling from the subset.

As the number of purged rows is not fully predictable, it can be necessary to generate more samples \REV{than} expected. To avoid this when possible, one can choose a sufficiently large sample. A common source for a large number of distinct values are join columns that are database keys. In that particular case, one can predict:

\begin{lemma} \label{lem:hash}
Let $T_1, T_2, \ldots, T_k$ be $k$ tables with unique values along the join columns.
The hash-relaxed equi-join $T_1 \bowtie_u T_2 \bowtie_u \ldots \bowtie_u T_k$ using a universal hash function with universe of size $u$ is expected to have at most $2m {(\frac{m}{u})}^{k-1}$ superfluous results that are not present in $T_1 \bowtie T_2 \bowtie \ldots \bowtie T_k$ where $m = \max(|T_1|, |T_2|, \ldots, |T_k|)$.
\end{lemma}

Thus, as a heuristic, it is proposed to collect a $2 {(\frac{m}{u})}^{k-1}$ times larger sample for the hash-relaxed superset ($m$ and $u$ as defined in Lemma~\ref{lem:hash}) if that does not exceed the memory limit, because the join size is expected to be at least as large as the tables.  As the number $k$ counts the tables where $u$ is much smaller than $m$, it is dependent on the choice of $u$. Thus, different choices of $u$ can be tried out find the best choice within the formula before the sampling commences.
When the memory limit does not allow the needed sample size, or the join has many distinct values without joining keys, then the hashing algorithm can be run multiple times with different random seeds to collect the targeted amount of samples.

\newcommand{\countjoin}[2]{ \, {\overset{\pi_{#2}}{\bowtie}}_{#1} \, }

\newcommand{\PARENT}{\mathcal{P}}
\newcommand{\PARENTJ}{P}

\newcommand{\TABLE}{\mathcal{T}}
\newcommand{\TABLEJ}{T}

\newcommand{\TABLEBATCH}{T^{batch}}
\newcommand{\PARENTBATCH}{P^{batch}}

\newcommand{\CHILD}{\mathcal{C}}
\newcommand{\CHILDJ}{C}
\newcommand{\CIRCLED}[1]{{\large \textcircled{\small #1}}}

%\section{Weighted Random Sampling (General)}
%

\section{Online Multinomial Sampler}
\label{sec:general}

This section solves how to draw a multinomial sample for an unknown population size where multinomial probabilities are only given indirectly through weights that are proportional to the probabilities. Each population item is only seen once (one-pass) and only memory proportional to the sample size is available, such that storing all population items it not feasible.
The online multinomial sampling problem (equivalent to the online weighted with-replacement problem) has been alluded to in the literature \cite{efraimidis2015weighted}, but appears to not have been explicitly solved. Note that unlike in \cite{jayaram2019weighted} the weights here are real-valued and unlike in \cite{startek2016asymptotically} only weights proportional to probabilities are available. The baseline technique from the literature would be to maintain an independent sampler for each element of the with-replacement sample, which does not scale with larger sample sizes. This section proposes an online multinomial sampler based on an adaption of an existing sampling technique \cite{DBLP:journals/ipl/EfraimidisS06}.

\begin{algorithm}
\DontPrintSemicolon
Initiate $S_1, S_2, \ldots, S_n$ as weighted reservoir sample  \label{alg:initbegin} \;

Initiate $W_\mathcal{P}$ as $0$ and multiset $R$ as $\{\}$ \label{alg:initend} \;
\ForEach{stream item $x$ with weight $w(x)$}{

    Consider $x$ for inclusion in $S$ and add $w(x)$ to $W_\mathcal{P}$ \;
} \label{alg:streamend}
Initiate $W_M$ as $0$, $j$ as $1$ and $\ell$ as $1$\;
\ForEach{$j \in \{1, \ldots, n\}$}{
    
    Draw $u \in [0,1]$ with $Pr(u) \propto 1$ \label{alg:coinflip} \;
    
    \If{ $u < W_M/W$ \label{alg:coinflipend}}{
        Set $M_j$ to a randomly drawn $M_i \in \{M_1, \ldots, M_{i-1} \}$ with $Pr(M_i) \propto w(M_i)$ \;
    } \Else {
        Add $w(S_\ell)$ to $W_M$, set $M_j$ to $S_\ell$ and set $\ell$ to $\ell+1$\;
    }
}\label{alg:end}
Return multinomial sample $M_1, \ldots, M_n$ \;
\caption{Online Multinomial Sampler \label{alg:onlinesampler}}

\end{algorithm}

The basic idea of the proposed approach is to draw an independent element with probabilities proportional to weights at each step. 
Observe that after the first item is drawn, the probability of drawing one of the previous items again is simply proportional to the weights of those previous items.  Thus, one needs to consider two cases. In one case the previously selected item is drawn again and in the other case a new random element is selected from the remaining population without previously selected items.  For the first case only need the previously selected items are needed and for the second case only items are needed from an equally small ordered sample that can be obtained using existing techniques \cite{DBLP:journals/ipl/EfraimidisS06} which will serve as a proxy for the population items.

Ordered sampling brings the population items into an order such that one can pick the first (or last) items as the sample. Using this idea, one can generate for each item with weight $w_i$ an independent exponential variate $v_i \sim Exp(w_i)$ and order them by these variates. It follows then from well-known properties of exponential variates that $Pr(v_i = \min( \{v_1, v_2, \ldots, v_N\}) = \frac{w_i}{w_1+w_2+\ldots+w_N}$. Thus, the minimum is chosen with probability proportional to weights, the second-smallest item with probability proportional to weights if the minimum item did not exist, and so on.  Efraimidis \& Spirakis \cite{DBLP:journals/ipl/EfraimidisS06} (E\&S) analogously use $k_i = e^{-v_i} \sim U(0,1)^{1/w_i}$ as random ``keys'', which flips the order by applying the strictly decreasing function $e^{-x}$, but otherwise yields the same properties and probabilities just for the maximum. E\&S then define the $n$ items with the largest $k_i$ as a ``weighted sample'' of size $n$, although only the largest-key item is chosen proportional to weights in the final sample \cite{tille2019general}.  Thus, E\&S do not use the weights to decide probabilities of the overall sampling process, but only in relation to rounds of without-replacement draws in the urn model. 
Meanwhile, here the $n$ largest-key items are only used as an intermediary for the population. Weighted with-replacement sampling has a very simple interpretation, i.e., each item is independently selected into the sample with probability proportional to its weight.

Now, the proposed approach can be described in Algorithm~\ref{alg:onlinesampler}. The first step of the proposed approach is to find the total weight $W_\mathcal{P}$ of all population items $\mathcal{P}$ and the $n$ largest-key items $S_1, S_2, \ldots, S_n$ from $\mathcal{P}$ still sorted by the keys, i.e., $k_1 \ge k_2 \ge \ldots \ge k_n$ (Lines~\ref{alg:initbegin}-\ref{alg:streamend}).  From that, it follows that all other population items $\mathcal{P} /\{S_1, S_2, \ldots, S_n\}$ have a smaller key than $k_n$ and will not make it into the multinomial sample. To find the largest-key items one can employ the weighted reservoir sampler by E\&S \cite{DBLP:journals/ipl/EfraimidisS06} with the exponential jump algorithm. In the next steps of the proposed approach (Lines~\ref{alg:coinflip}-\ref{alg:end}), one draws in each step an independent element $M_j$ with $j \ge 1$ and maintains at each moment the total weight $W_M$ of all distinct elements $\{M_1, M_2, \ldots, M_{j-1}\}$ that were previously selected. Note that for $M_j = M_1$ the ``abuse of notation'' $\{M_1, M_2, \ldots, M_{j-1}\}$ is used to denote an empty set and for $M_j = M_2$ it is equal to $\{M_1\}$. With probabilities proportional to weights, $M_j$ is either in the set of previously chosen item or a previously unchosen item:
$$
M_j \in \begin{cases} \{M_1, M_2,\ldots,M_{j-1}\} &  \text{with probability } W_M/W_\mathcal{P} \\ \mathcal{P} / \{M_1, M_2,\ldots,M_{j-1}\} & \text{otherwise} \end{cases}
$$

A (biased) coin flip determines according to these probabilities from which set one random element is selected with probabilities according to weights. Observe that it is exploited here that $\{S_1\}$ is a one-item weighted sample from $\mathcal{P}$, $\{S_2\}$ is a one-item weighted sample from $\mathcal{P} / \{S_1\}$,  $\{S_3\}$ is a one-item weighted sample from $\mathcal{P} / \{S_1, S_2\}$ and so on. All these one-item samples are also independent from each other, since they are selected based on independent keys. Thus, this algorithm offers a simple way to obtain a multinomial sample in a challenging setting of an unknown population where just one stream pass is allowed.

\begin{figure}\centering
\begin{minipage}[b]{.475\linewidth} \centering  \begin{tikzpicture}[scale=2.5]
\draw (0, 0.0) -- (-0.025, 0.0) node[anchor=east] {\tiny $0.0$};
\draw (0, 0.2) -- (-0.025, 0.2) node[anchor=east] {\tiny$0.2$};
\draw (0, 0.4) -- (-0.025, 0.4) node[anchor=east] {\tiny$0.4$};
\draw (0, 0.6) -- (-0.025, 0.6) node[anchor=east] {\tiny$0.6$};
\draw (0, 0.8) -- (-0.025, 0.8) node[anchor=east] {\tiny$0.8$};
\draw (0, 1.0) -- (-0.025, 1.0) node[anchor=east] {\tiny$1.0$};
\draw (-0.1, 1.1) node[anchor=west] { $Pr(X \le a)$};
% = \sum_{i = 1}^{a} p_i
\newcommand{\xxa}{0.0}
\newcommand{\xxb}{0.1}
\newcommand{\xxc}{0.2}
\newcommand{\xxd}{0.3}
\newcommand{\xxe}{0.4}
\newcommand{\xxf}{0.5}
\newcommand{\xxg}{0.6}
\newcommand{\xxh}{0.7}
\newcommand{\xxi}{0.8}
\newcommand{\xxj}{0.9}
\newcommand{\xxl}{1.0}
\newcommand{\yya}{0.05}
\newcommand{\yyb}{0.2}
\newcommand{\yyc}{0.4}
\newcommand{\yyd}{0.45}
\newcommand{\yye}{0.475}
\newcommand{\yyf}{0.5}
\newcommand{\yyg}{0.6}
\newcommand{\yyh}{0.8}
\newcommand{\yyi}{0.95}
\newcommand{\yyj}{0.978}
\newcommand{\yyl}{1.0}
\newcommand{\xxs}{1.0}
\newcommand{\xxt}{1.2}
\draw[draw=red, ultra thin] (\xxa,\yya) -- (\xxb, \yyb) --  (\xxc, \yyc) -- (\xxd, \yyd)--  (\xxe, \yye) --  (\xxf, \yyf) -- (\xxg, \yyg) --  (\xxh, \yyh) --  (\xxi, \yyi) --  (\xxj, \yyj)  -- (\xxl, \yyl) -- (1.0,1.0) -- (1.0, 0.0) -- (0.0, 0.0) -- cycle;
\draw[fill=black!10] (\xxa,\yya) -- (\xxb, \yya) -- (\xxb, \yyb) --  (\xxc, \yyb) -- (\xxc, \yyc) -- (\xxd, \yyc)  -- (\xxd, \yyd)-- (\xxe, \yyd) -- (\xxe, \yye) -- (\xxf, \yye) -- (\xxf, \yyf) -- (\xxg, \yyf) -- (\xxg, \yyg) -- (\xxh, \yyg) -- (\xxh, \yyh) -- (\xxi, \yyh) -- (\xxi, \yyi) -- (\xxj, \yyi) -- (\xxj, \yyj)  -- (\xxl, \yyj) -- (\xxl, \yyl) -- (1.0,1.0) -- (1.0, 0.0) -- (0.0, 0.0) -- cycle;
\draw (1, 0) node[anchor=west] { $a \in \mathbb{N}$};
\draw[thick] (0,0) rectangle (1,1);
\end{tikzpicture} \subcaption{Discrete Distribution \label{fig:convdisc}} \end{minipage}
\begin{minipage}[b]{.475\linewidth} \centering \begin{tikzpicture}[scale=2.5]
\draw (0, 0.0) -- (-0.025, 0.0) node[anchor=east] {\tiny $0.0$};
\draw (0, 0.2) -- (-0.025, 0.2) node[anchor=east] {\tiny$0.2$};
\draw (0, 0.4) -- (-0.025, 0.4) node[anchor=east] {\tiny$0.4$};
\draw (0, 0.6) -- (-0.025, 0.6) node[anchor=east] {\tiny$0.6$};
\draw (0, 0.8) -- (-0.025, 0.8) node[anchor=east] {\tiny$0.8$};
\draw (0, 1.0) -- (-0.025, 1.0) node[anchor=east] {\tiny$1.0$};
\draw (-0.1, 1.1) node[anchor=west] { $Pr(X \le x)$};
% = (\sum_{i = 1}^{\lceil x \rceil} p_i)-p_{\lceil x \rceil}(\lceil x \rceil-x)
\newcommand{\xxa}{0.0}
\newcommand{\xxb}{0.1}
\newcommand{\xxc}{0.2}
\newcommand{\xxd}{0.3}
\newcommand{\xxe}{0.4}
\newcommand{\xxf}{0.5}
\newcommand{\xxg}{0.6}
\newcommand{\xxh}{0.7}
\newcommand{\xxi}{0.8}
\newcommand{\xxj}{0.9}
\newcommand{\xxl}{1.0}
\newcommand{\yya}{0.05}
\newcommand{\yyb}{0.2}
\newcommand{\yyc}{0.4}
\newcommand{\yyd}{0.45}
\newcommand{\yye}{0.475}
\newcommand{\yyf}{0.5}
\newcommand{\yyg}{0.6}
\newcommand{\yyh}{0.8}
\newcommand{\yyi}{0.95}
\newcommand{\yyj}{0.978}
\newcommand{\yyl}{1.0}
\newcommand{\xxs}{1.0}
\newcommand{\xxt}{1.2}
\draw[fill=black!10] (\xxa,\yya) -- (\xxb, \yyb) --  (\xxc, \yyc) -- (\xxd, \yyd)--  (\xxe, \yye) --  (\xxf, \yyf) -- (\xxg, \yyg) --  (\xxh, \yyh) --  (\xxi, \yyi) --  (\xxj, \yyj)  -- (\xxl, \yyl) -- (1.0,1.0) -- (1.0, 0.0) -- (0.0, 0.0) -- cycle;
\draw[draw=red, ultra thin] (\xxa,\yya) -- (\xxb, \yya) -- (\xxb, \yyb) --  (\xxc, \yyb) -- (\xxc, \yyc) -- (\xxd, \yyc)  -- (\xxd, \yyd)-- (\xxe, \yyd) -- (\xxe, \yye) -- (\xxf, \yye) -- (\xxf, \yyf) -- (\xxg, \yyf) -- (\xxg, \yyg) -- (\xxh, \yyg) -- (\xxh, \yyh) -- (\xxi, \yyh) -- (\xxi, \yyi) -- (\xxj, \yyi) -- (\xxj, \yyj)  -- (\xxl, \yyj) -- (\xxl, \yyl) -- (1.0,1.0) -- (1.0, 0.0) -- (0.0, 0.0) -- cycle;
\draw (1, 0) node[anchor=west] { $x \in \mathbb{R}_+$};
\draw[thick] (0,0) rectangle (1,1);
\end{tikzpicture} \subcaption{Continuous Distribution \label{fig:conversioncont}} \end{minipage}
\caption{Converting discrete distribution into continuous one by adding uniform variates (cf. Lemma~\ref{lem:conv}) to apply continuous goodness-of-fit tests such as the K-S Test \label{fig:conv} } 
\end{figure}
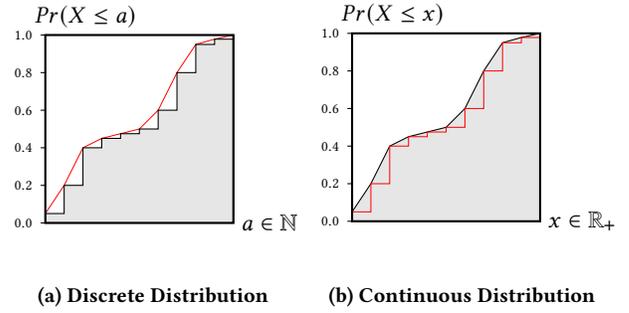

\section{ Goodness-of-Fit Tests for Multinomial Distributions}
\label{sec:gof}

In this section it is shown how to apply Kolmogorov-Smirnov (KS) testing when the reference distribution is discrete, as an alternative to goodness-of-fit tests for multinomial distributions such as the likelhood ratio test\cite{hoeffding1965asymptotically, cressie1984multinomial}.
Multinomial tests ignore the ordering unlike KS-tests that were originally designed for continuous distributions. There exist discrete variants of KS-tests \cite{pettitt1977kolmogorov,arnold2011nonparametric}, but they lose a lot of desirable properties of the continuous one such as distribution-free statistics, which is why it is here instead proposed to use the (conventional) continuous KS-test after turning the discrete distribution into a continuous one in a straightforward way:

\begin{lemma}[continuous conversion]\label{lem:conv}
Let $[x_1, x_2, \ldots, x_N]$ be a multinomially distributed frequency vector with event probabilities $p_1, \ldots, p_N$. Let $\mathcal{X}$ be a set obtained by drawing for the $i$-th event $x_i$ independent uniform variates between $(i-1)$ and $i$. Then $\mathcal{X}$'s distribution is continuous with a piecewise-linear cumulative distribution function (CDF), i.e., $F(x) = p_{\lceil x \rceil} (x-\lceil x \rceil)+\sum_{i = 1}^{\lceil x \rceil} p_i$.  
\end{lemma}
\begin{proof}
As shown in Figure~\ref{fig:conv}, the steps of the discrete CDF are smoothed by adding uniform variates, which have a continuous uniform distribution with a linear CDF. 
As a result one obtains a CDF that interpolates linearly between the steps without any jumps.
\end{proof}
The suggested continuous conversion describes how both samples and reference distributions have to be modified and is a clean way of applying continuous goodness-of-fit tests. While this comes with the disadvantage of slightly diluting the original distribution with additional randomness, this is needed to extend the distribution's finite support to an infinite one. Luckily, the impact of the added randomness quickly dissipates with larger discrete support size in this application.

\section{Related Work}
\label{sec:related}

All previous works on join sampling are proposed for simple random sampling (equi-weighted) or employ heuristics to sample from static joins \cite{huang13joins, DBLP:journals/pvldb/HadjieleftheriouYKS08, DBLP:conf/sigmod/KandulaSVOGCD16,DBLP:conf/sigmod/ParkMSW18, DBLP:conf/sigmod/HaasH99,DBLP:conf/cidr/LeisRGK017, DBLP:conf/sigmod/ChenY17, DBLP:conf/sigmod/0001WYZ16, DBLP:conf/eurosys/AgarwalMPMMS13} or joins of streams \cite{feraud2010sampling, DBLP:conf/ssdbm/Al-KatebLW07a, ren2007random,DBLP:series/ads/XieY07}. Samples collected using heuristics do not follow any well-defined distribution.
In what follows simple random sampling techniques are discussed (see a more detailled discussion in Section~\ref{sec:joinsamplingbackground}), which are a special case of weighted sampling with equal weights.

Foreign key joins that reunite tables from a normalised schema are equivalent to sampling from one table and then extending the sample using the other tables \cite{DBLP:conf/sigmod/AcharyaGPR99a,DBLP:conf/ssdbm/GemullaRL08}. As a result, it is sufficient to maintain one random sample, i.e., a join synopsis. 
In essence it is therefore more akin to conventional rather than join sampling. 
Interestingly, if weights or selections are introduced these approaches face new challenges that are not so easily addressed. 

In \cite{DBLP:conf/sigmod/Zhao0L20} a join sample is maintained for one particular $\theta$-join query over an evolving data stream (akin to monitoring queries), which means interactive user queries are not supported.

\REV{
Acyclic joins have a tree-like structure where each node corresponds to a tuple's (sub)join with the tables downstream, which allows decompositions akin to Yannakakis's work on acyclic conjunctive queries \cite{DBLP:conf/vldb/Yannakakis81}. This is useful when reasoning about join sampling algorithms.
Zhao \emph{et al.} \cite{DBLP:conf/sigmod/ZhaoC0HY18} were first to uncover the constraint that the extension of a uniformly sampled tuple \emph{must} be drawn with probability proportional to the size of the (sub)join of a tuple with the tables downstream, which generalises the insights from \cite{DBLP:conf/sigmod/ChaudhuriMN99} from binary joins to multi-way joins.
The key idea of the algorithm in \cite{DBLP:conf/sigmod/ZhaoC0HY18} is a tuple-oriented approach that first approximates this probability and then later rectifies it via rejection sampling, which is inspired by Olken's method \cite{olken1993random}. Rejections are avoided if approximations are exact. As the approach extends each tuple individually table-by-table, it mandates an index for efficient operation as explained by the authors of \cite{DBLP:conf/sigmod/ZhaoC0HY18} in the ``Remarks and Extensions'' section}\footnote{\REV{Quote from \cite{DBLP:conf/sigmod/ZhaoC0HY18}: "as with most random sampling techniques, our random sampling framework relies on efficient random access to the underlying data. Hence, we do require index on the join attribute when going from one relation to another. This means that our framework is mostly useful for an in-memory database where index structures are often built and maintained for join attributes."}}.
This was the first work on multi-way join sampling, and it reinvigorated interest in random sampling over joins after a long dormant period. The approach of \cite{DBLP:conf/sigmod/ZhaoC0HY18} to cyclic joins relies heavily on breaking the cycles near-optimally, but it is not clearly described how this optimisation problem can be solved efficiently. For sample validation, Zhao \emph{et al.} \cite{DBLP:conf/sigmod/ZhaoC0HY18} propose to utilise conventional KS-testing to validate samples over join rows. This is at odds with the statistics literature on (conventional) KS-testing, which requires a \emph{continuous} reference distribution\footnote{While KS-testing can be adopted to the discrete setting \cite{pettitt1977kolmogorov,arnold2011nonparametric}, the work \cite{DBLP:conf/sigmod/ZhaoC0HY18} reports distribution-free critical values which do not exist in the discrete case. The range of the reported $D$-statistics would also (almost certainly) not occur for the reported sample sizes, if the join result had sufficiently many distinct results to approximate a continuous distribution.}, and appears to confuse type-I and type-II errors. 
%In the same section, Type-I and type-II errors are mixed up, which further adds to the confusion of the reader. 
Hence, in Section~\ref{sec:gof}, a simpler and correct approach to goodness-of-fit testing is proposed. 
% show in this work how these issues can be addressed correctly (cf. ). 
The experimental evaluation compares against sampling over a materialized join, and assumes that all data resides in memory, including pre-computed indices. The results in this work show that index-free solutions based on scans are not much slower, and are well-known to be more amenable to parallelisation.  Critically, none of discussed works support weighted random sampling over joins.

The major takeaway is that while existing solutions such as \cite{DBLP:conf/sigmod/ZhaoC0HY18,shanghooshabad2021pgmjoins} solve the problem for equal-probability sampling over inner joins, but do so at the cost of expensive indexing techniques, while in this work index-free methods are proposed that offer substantial memory savings that can also support selections, outer joins, semi/anti-joins and operate over streams, which is crucial for heterogeneous data sources.

%Many works compute random samples over windowed streams of items \cite{DBLP:conf/soda/BabcockDM02, DBLP:conf/pods/BravermanOZ09, DBLP:journals/vldb/TaoHQ17}. For random samples over windowed stream joins used memory cannot be smaller than the window size \cite{DBLP:conf/vldb/SrivastavaW04}.

%Random samples over joins of non-stream tables \cite{DBLP:journals/pvldb/HadjieleftheriouYKS08, DBLP:conf/sigmod/KandulaSVOGCD16,DBLP:conf/sigmod/ParkMSW18, DBLP:conf/sigmod/ChaudhuriMN99,DBLP:conf/sigmod/ZhaoC0HY18}.

%Join of continuous stream with non-retroactive relations \cite{DBLP:conf/sigmod/GolabO05} (stream tuples joined with relations in current state, as changes to relations would require another pass over stream)

%Survey of works on stream joins \cite{DBLP:series/ads/XieY07}.

%Joins over windowed streams \cite{DBLP:conf/vldb/GolabO03, DBLP:conf/icde/KangNV03, DBLP:conf/ssdbm/HammadAE03, DBLP:conf/icde/HammadMAACEEEGGIMX04, DBLP:conf/sigmod/AnanthanarayananBDGJQRRSV13, DBLP:journals/jiis/KimPCK14}.

%Load shedding over streams \cite{DBLP:conf/icde/BabcockDM04,DBLP:conf/sigmod/DasGR03, DBLP:series/ads/BabcockDM07, DBLP:conf/vldb/TatbulCZ07, DBLP:books/sp/16/ArasuBBCDIMSW16}.

%Load shedding for windowed stream joins \cite{DBLP:conf/vldb/SrivastavaW04, DBLP:conf/icde/LawZ07, DBLP:journals/jcst/HanWXZ07}.

\section{Experiments}
\label{sec:experiments}

This section compares the novel stream sampler (minimising scans) and economic sampler (minimising memory) to existing approaches.
Very large numbers that may be troublesome in practice are highlighted in red to draw the attention of the reader.

All experiments are performed on a single thread of a dedicated machine using Ubuntu 18.04.4 LTS, an Intel(R) Xeon(R) W-2145 CPU @ 3.70GHz with 16 cores and 512GB RAM. Memory measurements are taken using the unix primitive {\tt/usr/bin/time -v} that provides the ``maximum resident set size". All code has been written in C++11 by the same author and was compiled using GCC 7.5.0 with the compiler flags {\tt-O3} and {\tt-std=c++11}.

\subsection{Join Queries}

\begin{table}
\begin{tabular}{rrrrr}
\toprule
TPC-H        &  ${}_{SF=10}$  & ${}_{SF=100}$ & \multicolumn{2}{l}{Real-World Data Sets} \\ %& ${}_{SF=1}$ 
 \midrule
%(W)Q3  &  59986052 & 600037902 \\ % &6000003
(W)Q3  &  $6.0 \cdot {10}^{7}$ & $6.0 \cdot {10}^{8}$ & DBLP & $4.5 \cdot {10}^{7}$\\ % &6000003
%(W)QX &  239942882657 &  24001463838651\\ % &  2447505622
(W)QX &  $2.4 \cdot {10}^{\color{red}12}$ &  $2.4 \cdot {10}^{\color{red}13}$ & Twitter QT & $4.8 \cdot {10}^{\color{red}10}$\\ % &  2447505622 
%(W)QY  &  527498410881 &  52537528389169\\ %& 4608989031142 
(W)QY  &  $5.3 \cdot {10}^{\color{red}11}$ &  $5.2 \cdot {10}^{\color{red}13}$ & Twitter QF & $2.7 \cdot {10}^{\color{red}21}$\\ %& 4608989031142
%\midrule
%DBLP  & Twitter QT  & Twitter QF  \\
%\midrule
%45564150 &   $\approx$ 48683941756 & 2705940129599796543488 \\
%$4.5 \cdot {10}^{7}$ &  $4.8 \cdot {10}^{10}$ & $2.7 \cdot {10}^{21}$ \\
\bottomrule
\end{tabular}
\caption{Join sizes of queries used in the experiments \label{tab:joinsizes}. Q3 and DBLP are foreign-key joins, QX and QF are acyclic many-to-many joins and QY and QT are cyclic joins.}
\end{table}

The experiments feature join queries over three datasets with join sizes shown in Table~\ref{tab:joinsizes}. The dataset used in the TPC-H benchmark, a social network of twitter users and a citation network using records from DBLP. Over the TPC-H benchmark the same queries are used as in \cite{DBLP:conf/sigmod/ZhaoC0HY18} (cf. Figure~\ref{fig:tpch}), but additionally define weights. The scale factor corresponds roughly to the size of the dataset in GBs. As a user-defined weighting function {\tt{o\textunderscore{}totalprice}} (1-{\tt{l\textunderscore{}discount}}) {\tt{l\textunderscore{}extendedprice}} is used. For QY, the values of both instances of {\tt{lineitem}} and {\tt{order}} are multiplied with each other. Weighted queries are referred to as WQ3 (foreign-key join), WQX (many-to-many join) and WQY (cyclic join). Over the (raw) twitter dataset \cite{DBLP:conf/www/KwakLPM10} the snowflake query QF and the triangle query QT is used as posed in \cite{DBLP:conf/sigmod/ZhaoC0HY18} (cf. Table~\ref{tab:expacyclic} and Table~\ref{tab:expcyclic}). For the DBLP data from Arnet Miner \cite{DBLP:conf/kdd/TangZYLZS08} a sample is collected over the citations%($N > {10}^{7}$)
, weighted by the year in which the citing paper was published (cf. Figure~\ref{fig:expweights}).

\begin{figure}[t]
\centering
     \begin{subfigure}[b]{0.1\textwidth}
     \centering
         \begin{tikzpicture}[yscale=0.3,node distance=1.25cm]
\tikzstyle{relation} = [rectangle, draw=black]
\tikzstyle{attribute} = []
\draw node(l)[relation]{lineitem $l$};
\draw node(o)[relation, below of=l]{order $o$};
\draw node(c)[relation, below of=o]{customer $c$};
\path (c) -- (o) node[midway] (ck){\tt custkey};
\path (o) -- (l) node[midway] (ok){\tt orderkey};
\draw (c)--(ck);
\draw (o)--(ck);
\draw (o)--(ok);
\draw (l)--(ok);
\end{tikzpicture}
         \caption{Q3   \label{fig:q3}}
     \end{subfigure}
     \begin{subfigure}[b]{0.1\textwidth}
     \centering
         \begin{tikzpicture}[yscale=0.3,node distance=1.25cm]
\tikzstyle{relation} = [rectangle, draw=black]
\tikzstyle{attribute} = []

\draw node(n)[relation]{nation $n$};
\draw node(s)[relation, above of=n]{supplier $s$};
\draw node(c)[relation, above of=s]{customer $c$};
\draw node(o)[relation, above of=c]{order $o$};
\draw node(l)[relation, above of=o]{lineitem $l$};

\path (n) -- (s) node[midway] (nk1){\tt nationkey};
\path (s) -- (c) node[midway] (nk2){\tt nationkey};
\path (c) -- (o) node[midway] (ck){\tt custkey};
\path (o) -- (l) node[midway] (ok){\tt orderkey};

\draw (n)--(nk1);
\draw (s)--(nk1);
\draw (s)--(nk2);
\draw (c)--(nk2);
\draw (c)--(ck);
\draw (o)--(ck);
\draw (o)--(ok);
\draw (l)--(ok);
\end{tikzpicture}
         \caption{QX \label{fig:qx}}
     \end{subfigure}\begin{subfigure}[b]{0.2\textwidth}
         \centering
         \begin{tikzpicture}[yscale=0.3,node distance=1.25cm]
\tikzstyle{relation} = [rectangle, draw=black]
\tikzstyle{attribute} = []

\draw node(s)[relation]{supplier $s$} ;

\draw node(c1)[above of=s,relation,anchor=east,left=0.1cm]{customer $c_1$} ;
\draw node(c2)[above of=s,relation, anchor=west,right=0.1cm]{customer $c_2$} ;

\draw node(o1)[above of=c1, relation]{order $o_1$} ;
\draw node(o2)[above of=c2, relation]{order $o_2$} ;

\draw node(l1)[above of=o1, relation]{lineitem $l_1$} ;
\draw node(l2)[above of=o2, relation]{lineitem $l_2$} ;

\path (s) -- (c1) node[midway](nk1){};
\path (s) -- (c2) node[midway](nk2){};

\path (nk1) -- (nk2) node[midway](nk){\tt nationkey};

\draw (s) -- (nk);
\draw (nk) -- (c1);
\draw (nk) -- (c2);

\path (c1) -- (o1) node[midway](ck1){\tt custkey};
\draw (c1)--(ck1) -- (o1);
\path (c2) -- (o2) node[midway](ck2){\tt custkey};
\draw (c2)--(ck2) -- (o2);

\path (o1) -- (l1) node[midway](ok1){\tt orderkey};
\draw (o1)--(ok1) -- (l1);
\path (o2) -- (l2) node[midway](ok2){\tt orderkey};
\draw (o2)--(ok2) -- (l2);

\path (l1) -- (l2) node[midway,above=0.5cm](pk){\tt partkey};
\draw (l1)--(pk) -- (l2);

\end{tikzpicture}
        \caption{QY \label{fig:qy}}
         
     \end{subfigure}
\caption{Join queries over TPC-H data.  \label{fig:tpch} } 
\end{figure}
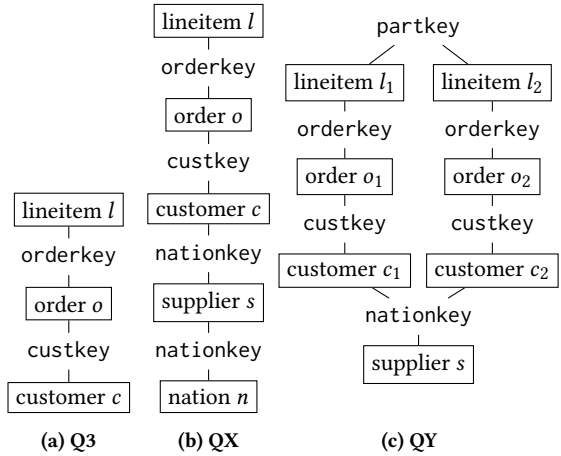

\subsection{Compared Approaches}

\begin{table}[]
    \centering
    \caption{Comparison of potential baselines}
    \begin{tabular}{llll}
\toprule
&  \multicolumn{3}{c}{\textbf{considered baselines}} \\
\midrule
&naive & \cite{DBLP:conf/sigmod/ZhaoC0HY18} code&index\\
\midrule
Q3$_{\ SF=100}$ memory&\cellcolor{yellow!20} 26.0 GB& \color{red} 142.0 GB&\color{red}  88.2 GB\\
Q3$_{\ SF=100}$ query-time&20.2 mins&\cellcolor{yellow!20}5.7 mins&7.1 mins\\
Q3$_{\ SF=100}$ adhoc-time&20.2 mins&64.2 mins&\cellcolor{yellow!20}18.2 mins\\
\midrule
QX$_{\ SF=100}$ memory& \color{red}  52.0 GB& \color{red} 169.7 GB&\color{red} 95.2 GB\\
QX$_{\ SF=100}$ query-time&\color{red} 1236.2 mins&\cellcolor{yellow!20}5.7 mins&7.9 mins\\
QX$_{\ SF=100}$ adhoc-time&\color{red} 1236.2 mins&91.5 mins&\cellcolor{yellow!20}20.3 mins\\
\midrule
QY$_{\ SF=10}$ memory&7.6 GB&\color{red} 36.1 GB&\color{red} 30.6 GB\\
QY$_{\ SF=10}$ query-time&23.4 mins&\cellcolor{yellow!20}5.0 mins&5.8 mins\\
QY$_{\ SF=10}$ adhoc-time&23.4 mins&14.4 mins&\cellcolor{yellow!20}9.8 mins\\
\midrule
QY$_{\ SF=100}$ memory& \color{red} 77.6 GB& \color{red} 384 GB &\color{red} 277.5 GB\\
QY$_{\ SF=100}$ query-time&\color{red} 895.1 mins& \cellcolor{yellow!20} 42 mins &52.9 mins\\
QY$_{\ SF=100}$ adhoc-time&\color{red} 895.1 mins& 116 mins &\cellcolor{yellow!20}102.7 mins\\
\bottomrule
\end{tabular}

    \label{tab:expbaselines}
\end{table}

\para{Naive.} 
This approach is a slightly more sophisticated version of first joining and then sampling. 
The approach joins together tables in a greedy fashion that reduces the sizes of the remaining tables until only one join column remains.
Then a merge-sort type of approach is used to retrieve the inverse of uniform variates. 
While this approach is not explicitly proposed in the literature, it seems fairly straightforward and is a clearly superior baseline to naive join-then-sample as it avoids materialising the full join result.

As shown in Table~\ref{tab:expbaselines} it is very slow and it is therefore not used as a general baseline.

\para{Index Baseline.} This approach attempts to implement \cite{DBLP:conf/sigmod/ZhaoC0HY18} with exact weights and generalises it to weighted sampling. 
As several details are not clear from the paper, many aspects are reverse engineered or simplified and generalised. 
The main goal of this implementation is to exploit indices in the same way as \cite{DBLP:conf/sigmod/ZhaoC0HY18}, but ideas from this work are used to fill in the gaps and streamline the approach. Note that while the authors of \cite{DBLP:conf/sigmod/ZhaoC0HY18} have made their code available, it is a bare bones research prototype that only supports integer-valued data, does not allow for non-join columns, is mostly optimised for speed without any way of outputting or validating samples. As Table~\ref{tab:expbaselines} shows, our code is comparably fast. Moreover, it can run on raw data, additionally supports weighted sampling and can be validated. Hence, the index-based implementation is used as the baseline to represent the state-of-the-art technique~\cite{DBLP:conf/sigmod/ZhaoC0HY18}.

\para{Stream (Proposed).}
The stream-approach implements the proposed approach from Section~\ref{sec:main} and prioritises a stream-like access over the data and limited number of scans for acyclic join queries. To only take one pass over the largest (or user-chosen) table it uses the  online multinomial sampler from Section~\ref{sec:general}.

\para{Economic (Proposed).}
The economic sampler implements the ideas from Section~\ref{sec:economic} and prioritises low memory usage. For foreign-key joins it first generates a uniform sample and then uses rejection sampling. For general acyclic joins, it uses the hashed-join technique to deal with high-cardinality attributes. For cyclic joins it attempts to simplify the join graph by greedily joining tables that have a small join size.

\subsection{Goodness-of-fit testing}

Figure~\ref{fig:expks} demonstrates that the author's implementations produce samples that follow the multinomial distributions based on the user-defined weights. This experiment uses a weighted cyclic query to make the test challenging, and also features approximations based on first sampling the base tables.
The Kolmogorov-Smirnov test statistic is obtained for each returned sample, and shading is used to indicate the critical region.
The considered approaches all stay with $99\%$ probability below the shaded region. 
For constrant, ``approximate'' solutions which evaluate the query based on samples of the tables (i.e., the ``sample then join'' approach) forare shown as well. 
Even very large samples (as large as 50\% of the base tables)
step into the shaded region, which means they have a statistically significant deviation from the target distribution.

\begin{figure}[t]
    \centering
    %!TEX root = test.tex
\begin{tikzpicture}\begin{loglogaxis}[height=7cm,width=8cm,grid=major,legend pos=south west,xlabel={sample size},xtick={1, 10, 100, 1000, 10000, 100000, 1000000},xticklabels={${10}^{0}$, ${10}^{1}$, ${10}^{2}$, ${10}^{3}$, ${10}^{4}$, ${10}^{5}$, ${10}^{6}$},ylabel={K-S-Test D-Statistic},ytick={1e-05, 0.0001, 0.001, 0.01, 0.1, 1},yticklabels={${10}^{-5}$, ${10}^{-4}$, ${10}^{-3}$, ${10}^{-2}$, ${10}^{-1}$, ${10}^{0}$},
xmin= 100, xmax= 1000000,ymin= 0.0001, ymax= 1,ymax=0.2,ymajorgrids=true,xmajorgrids=false]

\addplot[blue, mark=triangle*, every mark/.append style={solid, fill=blue!30},mark size=6pt] coordinates {
(100,0.0899279)
(1000,0.0163431)
(10000,0.0122011)
(100000,0.00155376)
(1000000,0.000612942)
};
\addlegendentry{index baseline}
\addplot[red,mark=square*, thick,every mark/.append style={solid, fill=red!30},mark size=5pt] coordinates {
(100,0.0731948)
(1000,0.0272541)
(10000,0.00626821)
(100000,0.00312304)
(1000000,0.00108899)
};
\addlegendentry{stream}
\addplot[red,mark=x, dashed, very thick,every mark/.append style={very thick,solid, fill=red!30},mark size=5pt] coordinates {
(100,0.0714585)
(1000,0.0344963)
(10000,0.0084093)
(100000,0.00208175)
(1000000,0.000650384)
};
\addlegendentry{economic}
%\addplot[black!70,very thick, mark=otimes, every mark/.append style={solid, fill=black},mark size=4pt] coordinates {
%(100,0.207297)
%(1000,0.142586)
%(10000,0.133279)
%(100000,0.134559)
%(1000000,0.13699)
%};
%\addlegendentry{approx${_{1\%}}$}
\addplot[black!70, very thick, mark=*, every mark/.append style={solid, fill=black!50},mark size=4pt] coordinates {
(100,0.0579021)
(1000,0.0244967)
(10000,0.0222946)
(100000,0.0074377)
(1000000,0.00808028)
};
\addlegendentry{approx${_{10\%}}$ baseline}
\addplot[black!70, very thick, mark=diamond*, every mark/.append style={solid, fill=white},mark size=5pt] coordinates {
(100,0.0984024)
(1000,0.0219261)
(10000,0.0134485)
(100000,0.00466278)
(1000000,0.00231134)
};
\addlegendentry{approx${_{50\%}}$ baseline}
\path[name path=axis] (axis cs:1,1) -- (axis cs:1000000,1);
\addplot[name path=f,domain=1:1000000,draw=none]  {1.62762/sqrt(x)};
\addplot[fill=black!20] fill between[of=f and axis];
\addplot[pattern=horizontal lines, pattern color=black!40] fill between[of=f and axis];
\path[name path=ax] (axis cs:1,0.0001) -- (axis cs:1000000,0.00001);
\addplot[name path=g,domain=1:1000000,draw=none]  {1.62762/sqrt(x)};
%\addplot[fill=white] fill between[of=g and ax];
\end{loglogaxis}
\end{tikzpicture}
    \caption{K-S goodness-of-fit-test (Section~\ref{sec:gof}) on cyclic WQY query $(SF=1)$ with over ${10}^{12}$ result rows.  } %4602524876868
    \label{fig:expks}
\end{figure}
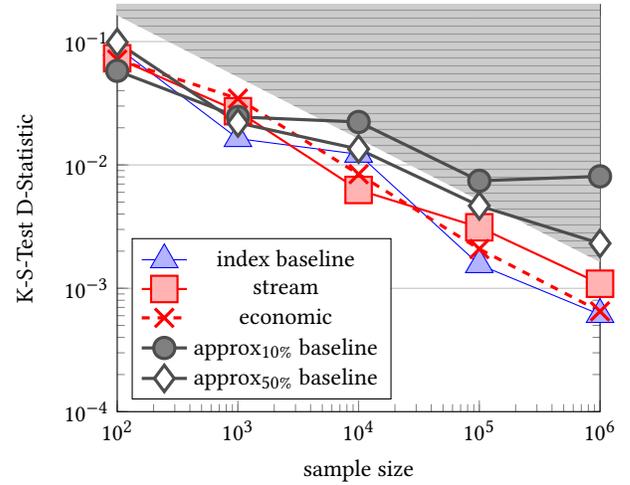

\subsection{Many-To-One Joins (Foreign-Key Joins)}

\begin{figure}[t]
    \centering
    %!TEX root = test.tex
\begin{tikzpicture}\begin{loglogaxis}[height=5cm,width=8cm,grid=major,legend pos=north west,xlabel={weighting function (with year of publication $Y$)},xtick={1.0, 1.1, 1.2, 1.3, 1.4, 1.5, 1.6, 1.7, 1.8, 1.9, 2.0},xticklabels={{${1.0}^{-Y}$}, {${1.1}^{-Y}$}, {${1.2}^{-Y}$}, {${1.3}^{-Y}$}, {${1.4}^{-Y}$}, {${1.5}^{-Y}$}, {${1.6}^{-Y}$}, {${1.7}^{-Y}$}, {${1.8}^{-Y}$}, {${1.9}^{-Y}$}, {${2.0}^{-Y}$}},ylabel={runtime (mins)},ytick={0.000167, 0.0167, 1, 10, 60, 600, 6000},yticklabels={{10 ms}, {1 sec}, {1 min}, {10min}, {1h}, {10h}, {100h}},
xmin= 1.0, xmax= 1.4,ymin= 0.1, ymax= 600,]

\addplot[red,mark=x, dashed, very thick,every mark/.append style={very thick,solid, fill=red!30},mark size=5pt] coordinates {
(1.0,0.901103366667)
(1.1,4.29688941667)
(1.2,6.3871623)
(1.3,16.8509422833)
(1.4,105.7490756)
};
\addlegendentry{economic}
\addplot[red,mark=square*, thick,every mark/.append style={solid, fill=red!30},mark size=5pt] coordinates {
(1.0,0.925435583333)
(1.1,0.66489425)
(1.2,0.570342333333)
(1.3,0.512777216667)
(1.4,0.471129)
};
\addlegendentry{stream}
\addplot[blue, mark=triangle*, every mark/.append style={solid, fill=blue!30},mark size=6pt] coordinates {
(1.0,1.12577918333)
(1.1,0.704796183333)
(1.2,0.586131716667)
(1.3,0.52198095)
(1.4,0.495857833333)
};
\addlegendentry{index baseline}
\end{loglogaxis}
\end{tikzpicture}
    \caption{Drastic impact of weighting on runtimes of a million samples over foreign-key joins on the DBLP dataset. }
    \label{fig:expweights}
\end{figure}
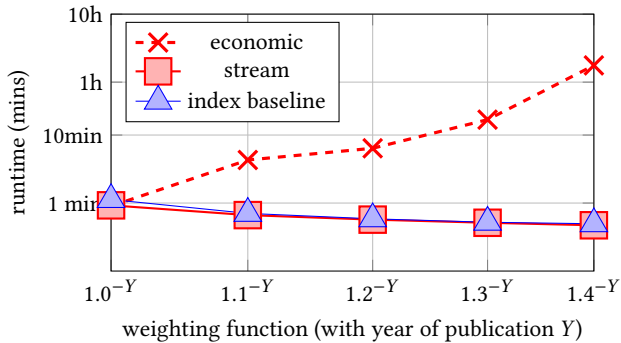

\begin{table}[b]
    \centering
    \caption{Runtime and memory comparison for a million samples over foreign-key joins}
    %!TEX root = test.tex
\begin{tabular}{ll|ll}
\toprule
& \textbf{baseline}& \multicolumn{2}{c}{\textbf{proposed}} \\
\midrule
& index &stream&economic\\
\midrule
WQ3$_{\ SF=10}$ memory&\color{red} 10.6 GB&\cellcolor{yellow!20}1.6 GB&1.9 GB\\
WQ3$_{\ SF=10}$ query-time&\cellcolor{yellow!20}0.7 mins&1.3 mins&1.4 mins\\
WQ3$_{\ SF=10}$ adhoc-time&1.8 mins&\cellcolor{yellow!20}1.3 mins&1.4 mins\\
\midrule
WQ3$_{\ SF=100}$ memory&\color{red} 94.6 GB&9.8 GB&\cellcolor{yellow!20}2.6 GB\\
WQ3$_{\ SF=100}$ query-time&7.9 mins&12.3 mins&\cellcolor{yellow!20}7.5 mins\\
WQ3$_{\ SF=100}$ adhoc-time&20.5 mins&12.3 mins&\cellcolor{yellow!20}7.5 mins\\
\bottomrule
\end{tabular}

    \label{tab:expfk}
\end{table}

For many-to-one joins, one can first sample from the ``many'' side of the relationships and will have exactly one extension on the other side. 
This makes the problem a lot easier to solve, but still poses some challenges in the weighted sampling case. 
To support weighted sampling, one can obtain an upper bound for the user-defined weights by computing the product of maximal base table weights and then accept uniformly sampled rows with probability equal to the ratio with the weight upper bound.
As can be seen in Table~\ref{tab:expfk}, this works well for linear weight functions. 
Here, the economic sampler uses less memory by assuming group weights being equal to one, which is then rectified through rejections. Anticipating some rejections, it collects a ten times larger sample.

However, for exponential weight functions as in Figure~\ref{fig:expweights} the approach becomes very slow as it has to repeatedly restart the process to collect more samples. In this example using the DBLP query, one sets the weights to be a varying exponential function based on the year of citations. Pre-computed indices do not seem beneficial in this case. It is therefore better to switch to the stream sampler if the rejection rate grows too large, as it does not appear to be affected by the weight distribution.

\subsection{Cyclic joins (For Sake of Completeness)}

While the main focus of this work are acyclic joins, results for cyclic joins are shown in in Table~\ref{tab:expcyclic}. Not shown in the table is QT due to reasons outlined in the following. For the cyclic QY query over the TPC-H dataset, the economic sampler can reduce the query to an acyclic join and thereby achieve some memory savings, but this is not possible for the triangle query QT over the Twitter dataset. 
For QT there exists an acyclic join that is not a lot larger than the cyclic join and the proposed samplers can efficiently find it. Still, the rejection rate is too large for anything but an index-based sampler to complete it in a competitive amount of time. The index-based sampler occupies for this task 454.8 GB of main memory and needs 341.7 mins from scratch and 277.5 mins for just the query-answering.

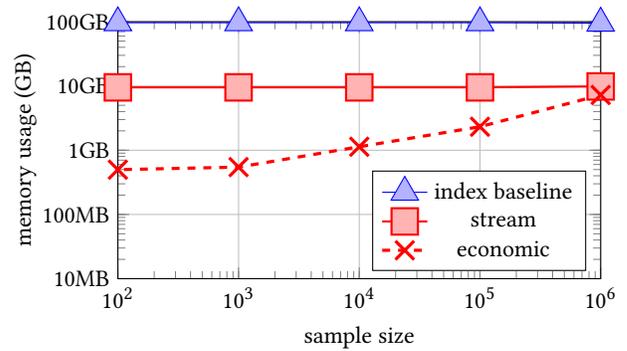
\begin{figure}[t]
    \centering
    %!TEX root = test.tex
\begin{tikzpicture}\begin{loglogaxis}[height=5cm,width=8cm,grid=major,legend pos=south east,xlabel={sample size},xtick={1, 10, 100, 1000, 10000, 100000, 1000000},xticklabels={${10}^{0}$, ${10}^{1}$, ${10}^{2}$, ${10}^{3}$, ${10}^{4}$, ${10}^{5}$, ${10}^{6}$},ylabel={memory usage (GB)},ytick={0.001, 0.01, 0.1, 1, 10, 100, 1000, 10000, 100000},yticklabels={{1MB}, {10MB}, {100MB}, {1GB}, {10GB}, {100GB}, {1TB}, {10TB}, {100TB}},
xmin= 100, xmax= 1000000,ymin= 0.01, ymax= 100,]

\addplot[blue, mark=triangle*, every mark/.append style={solid, fill=blue!30},mark size=6pt] coordinates {
(1,96.5300827026)
(10,96.5300636292)
(100,96.5300140381)
(1000,96.5300979614)
(10000,96.5299720764)
(100000,96.5301055908)
(1000000,95.0286102295)
};
\addlegendentry{index baseline}
\addplot[red,mark=square*, thick,every mark/.append style={solid, fill=red!30},mark size=5pt] coordinates {
(1,9.51320648193)
(10,9.51330566406)
(100,9.51321029663)
(1000,9.51473617554)
(10000,9.5132522583)
(100000,9.51337814331)
(1000000,9.83720016479)
};
\addlegendentry{stream}
\addplot[red,mark=x, dashed, very thick,every mark/.append style={very thick,solid, fill=red!30},mark size=5pt] coordinates {
(10,0.498146057129)
(100,0.498016357422)
(1000,0.546043395996)
(10000,1.13215255737)
(100000,2.32070922852)
(1000000,7.24256896973)
};
\addlegendentry{economic}
\end{loglogaxis}
\end{tikzpicture}
    \caption{Memory for different sample sizes over QX (SF=100)}
    \label{fig:exphash}
\end{figure}

\begin{table}[]
    \centering
    \caption{Runtime and memory comparison for a million samples over cyclic (many-to-many) joins.}
    %!TEX root = test.tex
\begin{tabular}{ll|ll}
\toprule
& \textbf{baseline}& \multicolumn{2}{c}{\textbf{proposed}} \\
\midrule
& index &stream&economic\\
\midrule
WQY$_{\ SF=10}$ memory&31.8 GB&28.5 GB&\cellcolor{yellow!20}11.7 GB\\
WQY$_{\ SF=10}$ query-time&\cellcolor{yellow!20}6.1 mins&18.9 mins&15.9 mins\\
WQY$_{\ SF=10}$ adhoc-time&\cellcolor{yellow!20}10.6 mins&18.9 mins&15.9 mins\\
\midrule
WQY$_{\ SF=100}$ memory& \color{red}290.1 GB&73.7 GB&\cellcolor{yellow!20}41.6 GB\\
WQY$_{\ SF=100}$ query-time&\cellcolor{yellow!20}55.9 mins&100.5 mins&68.4 mins\\
WQY$_{\ SF=100}$ adhoc-time&110.9 mins&100.5 mins&\cellcolor{yellow!20}68.4 mins\\
%\midrule
%QT memory&\cellcolor{yellow!20}454.8 GB&-&-\\
%QT query-time&\cellcolor{yellow!20}277.5 mins&-&-\\
%QT adhoc-time&\cellcolor{yellow!20}341.7 mins&-&-\\
\bottomrule
\end{tabular}

    \label{tab:expcyclic}
\end{table}

\subsection{General (Many-To-Many) Acyclic Joins}

\begin{table}[t]
    \centering
    \caption{Runtime and memory comparison for a million samples over acyclic (many-to-many) joins.}
    %!TEX root = test.tex
\begin{tabular}{ll|ll}
\toprule
& \textbf{baseline}& \multicolumn{2}{c}{\textbf{proposed}} \\
\midrule
& index &stream&economic\\
\midrule
WQX$_{\ SF=10}$ memory& \color{red} 10.7 GB&\cellcolor{yellow!20}1.6 GB&2.2 GB\\
WQX$_{\ SF=10}$ query-time&\cellcolor{yellow!20}0.7 mins&1.4 mins&1.8 mins\\
WQX$_{\ SF=10}$ adhoc-time&1.9 mins&\cellcolor{yellow!20}1.4 mins&1.8 mins\\
\midrule
WQX$_{\ SF=100}$ memory& \color{red} 95.0 GB&9.8 GB&\cellcolor{yellow!20}7.2 GB\\
WQX$_{\ SF=100}$ query-time&\cellcolor{yellow!20}7.9 mins&13.1 mins&14.0 mins\\
WQX$_{\ SF=100}$ adhoc-time&20.6 mins&\cellcolor{yellow!20}13.1 mins&14.0 mins\\
\midrule
QF memory& \color{red} 428.4 GB&91.5 GB&\cellcolor{yellow!20}16.2 GB\\
QF query-time&63.1 mins&\cellcolor{yellow!20}36.0 mins&47.0 mins\\
QF adhoc-time&126.5 mins&\cellcolor{yellow!20}36.0 mins&47.0 mins\\
\bottomrule
\end{tabular}

    \label{tab:expacyclic}
\end{table}

The query QX over the TPC-H dataset used in Figure~\ref{fig:exphash} does not have many distinct values, yet the economic sampler can vastly reduce the memory footprint if one only needs to collect a smaller sample, using the hashed-join technique.

The intuition is that for small samples one can choose to collect a larger set of tuples for the hashed-join technique, and then discard the false positives.

For the query QF over the Twitter dataset (cf. Table~\ref{tab:expacyclic}), one can also see vast memory savings for a million of sampled rows using the economic sampler method.
This is because it needs far less memory when the number of distinct join attribute values is large.

\section{Conclusion}

In this paper two new samplers for join queries were proposed that offer multiple novel features such as support for user-defined sampling weights and memory-efficient ad-hoc operation over stream-accessed tables, which is very useful as most data management systems offer some type of streaming interface to the data.
First, a stream-based sampler, which is the first join sampler that only needs stream access and can collect a join sample by passing once over one table and two times over the other tables. Second, an economical sampler, which uses various strategies to exploit foreign-key joins, simpler instances of cyclic joins and can handle high-cardinality join attributes through a novel hashed join technique that adheres to the targeted sampling probabilities without any approximations. The basic idea of replacing equi-joins with equi-hash joins to reduce the number of considered join attribute values is not present in the literature and offers here a crucial instrument to deal with common cases such as join attributes that are database keys. In the experiments it is shown that these new ideas translate to improved practical performance and show that indices are only useful cyclic joins  heuristics that need to generate a very large number of samples to find valid join rows. The open problems that remain are cyclic joins, non-factorisable weights with a very unequal distribution of weights, non-equi joins and other more complicated relational expressions.

\begin{acks}
\REV{This work is supported by European Research Council grant ERC-2014-CoG 647557 and The Alan Turing
Institute under the EPSRC grant EP/N510129/1.}
\end{acks}

\bibliographystyle{ACM-Reference-Format}
\bibliography{reference}

\end{document}